\documentclass[]{article}
\usepackage{amsmath,amsthm,amsfonts}
\usepackage{epstopdf}
\usepackage[caption=false]{subfig}
\usepackage{graphicx}
\usepackage{paralist}
\usepackage[numbers,sort&compress]{natbib}
\bibpunct[, ]{[}{]}{,}{n}{,}{,}
\renewcommand\bibfont{\fontsize{10}{12}\selectfont}
\makeatletter
\def\NAT@def@citea{\def\@citea{\NAT@separator}}
\makeatother

\theoremstyle{plain}
\newtheorem{theorem}{Theorem}[section]

\theoremstyle{definition}

\theoremstyle{remark}
\newtheorem{remark}{Remark}

\newcommand{\dnac}{{\delta_{\mathrm{n}}^{\mathrm{AC}}}}
\newcommand{\dndc}{{\delta_{\mathrm{n}}^{\mathrm{DC}}}}
\newcommand{\neff}{n_{\mathrm{eff}}}
\newcommand{\R}{\mathbb R}
\newcommand{\rmeas}{r^{\mathrm{meas}}}
\DeclareMathOperator*{\argmin}{arg\,min}

\usepackage{color}

\begin{document}


\title{Regularization for the inversion of Fibre Bragg Grating spectra}

\author{Daniel Gerth\footnote{Email: daniel.gerth@mathematik.tu-chemnitz.de,  \textsuperscript{a}Fakult\"at f\"ur Mathematik, Technische Universit\"at Chemnitz, 09107 Chemnitz, Germany, \textsuperscript{b}Fakult\"at f\"ur Maschinenbau, Technische Universit\"at Chemnitz, 09107 Chemnitz, Germany}\textsuperscript{a}, 
 Susann Hannusch\textsuperscript{b}, Oliver G.~Ernst\textsuperscript{a}, and\\ J\"{o}rn Ihlemann\textsuperscript{b}}

\maketitle

\begin{abstract}
Fibre Bragg Gratings have become widespread measurement devices in engineering and other  fields of application. In all but a few cases, the relation between cause and effect is simplified to a proportional model. However, at its mathematical core lies a nonlinear inverse problem which appears not to have received much attention in the literature. In this paper, we present this core problem to the mathematical community and provide a first report on opportunities and limitations of a regularization approach. In particular, we show that difficulties arise from non-uniqueness and the absence of established parameter selection rules for nonlinear inverse problems with multiple regularization parameters. Nevertheless, the paper takes a first step toward extracting more information from a single FBG measurement.
\end{abstract}

\section{Introduction}\label{sec:intro}

Since their first demonstration in the late 1970s \cite{HillEtAl70}, \emph{fibre Bragg gratings (FBGs)} have become a widely used technology for measuring physical quantities such as strain, temperature, or concentration. 
They also have a broad and important role in other settings, and we refer to \cite{HillMeltz1997} for an overview. 
Our intended application is the measurement of strain in fibre-reinforced polymers, but in this work we focus on a fundamental inverse problem common to all applications. 
Technologically, an FBG consists of a segment in an optical fibre in which a periodic variation of the refractive index has been inscribed.  
Some common FBG inscription patterns are shown in Figure~\ref{fig:FBGtypes}.
\begin{figure}  \centering
\includegraphics[width=\linewidth]{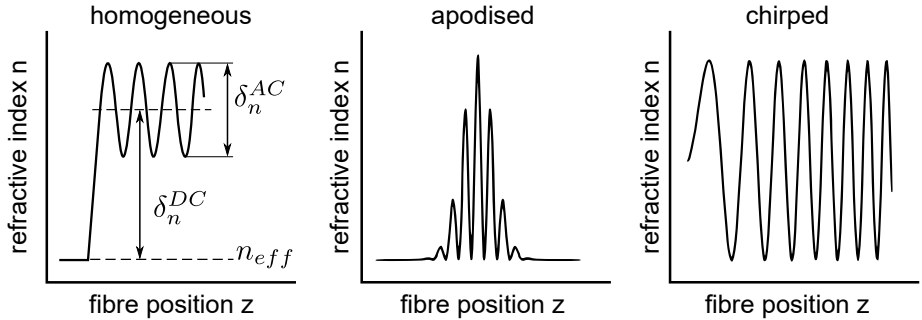}
\caption{Common refractive index ($n$) distributions along a fibre Bgragg grating.}
\label{fig:FBGtypes}
\end{figure}
Depending on the specific grating pattern, certain wavelengths of light are reflected at the grating as they propagate through the fibre. 
Hence, when the fibre carrying an FBG is probed with light, usually emanating from a tunable laser, a narrow band of wavelengths is reflected. 
The intensity of the reflected light (its \emph{spectrum}) serves as the measured data, and the fundamental task lies in recovering the refractive index distribution inside the grating zone from the observed spectrum. 
The quantity to be measured in the specific application is then inferred from the reconstructed refractive index distribution. 
A schematic view of an FBG and its measurements is shown in Figure~\ref{fig:FBGspectrum}. 

\begin{figure}  \centering
\includegraphics[width=0.9\linewidth]{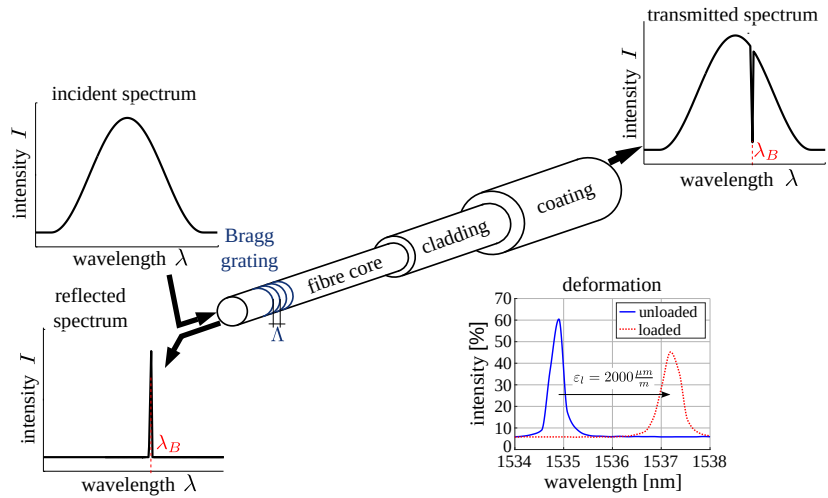}
\caption{Incident, reflected and transmitted spectra of light passing through a fibre Bragg grating sensor.}
\label{fig:FBGspectrum}
\end{figure}

In the majority of applications the spectrum consists of a single peak at a specific wavelength. 
Temperature changes or applied strain cause the peak to move to another wavelength. 
The difference between the old and new peak wavelengths allows a quantification of the applied strain (or other quantity of interest). 
It is well known that under certain conditions the single peak can broaden and even split into two (or even more) separate peaks. 
The simple formula for the peak deviation then becomes increasingly uninformative and inaccurate, and existing measurement devices then abort the measurement process or yield unrealistic data. 
Hence, peak deformation is widely regarded as an undesired phenomenon. 
However, such a multi-peak spectrum contains much more information besides the peak location. The aim of this work is to extend the analysis of FBG signals from mainly peak-focused considerations to the full spectrum, and to extract as much information from the spectrum as appears possible. 
We emphasize that this paper is written as an introduction of the problem to the inverse problems community and reports first results, but also leaves some open problems for future work.

The remainder of the paper is organized as follows. 
In Section~\ref{sec:fbg} we introduce basic FBG concepts as well as the model for the FBG refractive index distribution.
In Section~\ref{sec:forward} we sketch the derivation of the forward model to provide the physical background and discuss the implementation.
In Section~\ref{sec:unique} we provide a first result on the uniqueness of FBG spectra.
Our regularization approach is presented in Section~\ref{sec:reg}.
In Section~\ref{sec:experiments} we apply the theoretical results to experimental data.


\section{Fibre Bragg Gratings} \label{sec:fbg}

In the following we sketch the most important aspects of the physics underlying FBGs. 
Due to the complex nature of the topic, many details are omitted, and for these we refer to, for example, the monographs \cite{Kashyap1999,Marcuse1991,Okamoto2005}, or the survey papers \cite{HillMeltz1997,Erdogan1997} and references therein.

\par{\textbf{Optical waveguides.}}
As the name suggests, the base for an FBG is an optical fibre. 
Such fibres consist of three concentric layers: a relatively thin core with refractive index $n_\text{co}$ in which (most of) the light is guided, a comparatively thick cladding with refractive index $n_\text{cl}$, and a thin coating to protect the inner layers; see Figure \ref{fig:FBGspectrum}. 
Next we note that light is guided in a specific way in a fibre. 
Consider a plane wave travelling in the fibre at an incident angle $\theta$. 
Firstly, only incident angles $\theta\leq \sin^{-1}\left(\sqrt{n_\text{co}^2-n_\text{cl}^2}\right)$ can be guided, as otherwise the light is not reflected back into the core at the core-cladding interface, but rather  refracted into the cladding.
Secondly, even below this threshold, only a discrete set of inclination angles result in a propagating wave, as otherwise the wavefronts become out of phase and the wave breaks down. The superposition of the travelling wavefronts with constant phase in turn produces a standing wave inside the fibre. This standing wave is called a \emph{guided mode}. 
If the core refractive index $n_\text{co}$ and the cladding refractive index $n_\text{cl}$ are sufficiently close, then only a single mode is guided in the fibre.  
We will only consider this case in the following. 
We also mention that some modes can also be guided in the cladding, but this will not be pursued in this paper. 

\begin{remark}\label{rem:biref}
Another important effect one may encounter in measurements is \emph{birefringence}. 
This significantly adds to the complexity of the reconstruction, and in order to keep the presentation simple we ignore this effect in this paper. 
Birefringence results from the fact that each guided mode consists of two polarization states, which normally do not interact with each other. 
Mathematically speaking, the electric fields are orthogonal. 
Some external causes, most prominently a deformation of the fibre core from an initially circular cross section to an elliptical one, lead to differing propagation properties between the two polarization states. 
Orthogonality is lost, and the two polarization states can exchange energy. This is one of the main causes for the splitting of a the single peaks usually observed in FBGs. For more details see, e.g., \cite{Gafsi2000}. 
\end{remark}

The wavelength-dependent \emph{propagation constant} $\beta$ describes how phase and amplitude vary along the propagation direction. For a guided mode with incident angle $\theta$ travelling in a medium of refractive index $n_0$, it is given by
\[
	\beta = \frac{2\pi}{\lambda} n_0 \cos \theta.
\]
The \textit{effective refraction index} $\neff := n_0\cos\theta$ describes the propagation in the fibre direction. 
Because in a fibre $\theta$ is small (as $n_\text{co}-n_\text{cl} \ll 1$), $\neff$ and $n_\text{co}$ are often used (almost) interchangeably \cite{Erdogan1997}. 
One might say that $\neff$ is simply the average refraction index in the grating section of the fibre.

\par{\textbf{Bragg gratings}} are reflection gratings. 
The effect of a planar wave incident on a grating of (period-)length $\Lambda$ at an angle $\theta_1$ in a medium with refractive index $n_1$ can be described by the equation
\begin{equation}\label{eq:grating}
	n_1\sin \theta_1 = n_2 \sin\theta_2 + m\frac{\lambda}{\Lambda},
\end{equation}
where $n_2$ is the refractive index of the medium after refraction, $\theta_2$ the ``outgoing'' angle, $m\in \mathbb{N}$ the refraction order, and $\lambda$ the wavelength \cite{Born}. 
For Bragg gratings, the dominant refraction is first order ($m=-1$) refraction of a mode under its negative incident angle, i.e., $\theta_2=-\theta_1$. 
For a fixed mode with $n_1=n_2=\neff$, \eqref{eq:grating} can be written in the form of the well-known \emph{Bragg condition} for the design wavelength $\lambda_B$,
\begin{equation}\label{eq:gratingcondition}
	\lambda_{B} = 2\neff\Lambda.
\end{equation}
This is the basic relation which suffices as a basis for the measurement of strain, temperature, etc., in many applications. 
Specifically, a change in these quantities results in a shift of the  peak wavelength $\lambda_B$. 
For example, uniaxial strain $\varepsilon_z$ in fibre direction at constant temperature yields a peak shift of
\[
	d\lambda_B
	=
	2 \left[
	\left( \Lambda\frac{\partial \neff}{\partial \varepsilon_z} \right)
	+
	\neff\left(\frac{\partial \Lambda}{\partial \varepsilon_z} \right) 
	\right]
	d\varepsilon_z,
\]
see, e.g., \cite{Gafsi2000}. 
Using some engineering mechanics, one can reformulate this to the relation
\begin{equation}\label{eq:strain}
	\frac{\Delta \lambda_B}{\lambda_B} = K \varepsilon_z,
\end{equation}
i.e., the measured shift in peak wavelength $\Delta \lambda_B$ relative to the original peak wavelength is proportional ($K$ depends on material parameters) to the longitudinal strain.
This is a highly simplified model. The measurement devices (\textit{interrogators}) measure the intensity of the light reflected by the FBG over several wavelengths, the spectrum.  

However, \eqref{eq:gratingcondition} and the occurring quantities are only a rough characterization of the FBG. We now turn to a more detailed description. 
To this end, note that FBGs are usually inscribed into the fibre by phase-masking or via the interference of two UV-light rays (see, e.g., \cite{HillMeltz1997,Kashyap1999}). 
Although the fibre and thus the Bragg grating are three-dimensional objects, the refractive index distribution is modelled as one-dimensional, varying only along the longitudinal coordinate. 
Specifically, let the fibre section containing the FBG be of length $L>0$. 
Then the common model to describe the grating structure is
\begin{equation}\label{eq:dnreal}
	\delta_n(z)
	=
	\delta_n^\text{DC}(z) + \delta_n^\text{AC}(z) \cos\left(\frac{2\pi z}{\Lambda}+\Phi(z) \right), 
	\qquad -\frac{L}{2} \leq z\leq \frac{L}{2},
\end{equation}
where $z$ is the longitudinal coordinate. 
The functions $\delta_n^\text{DC}$ and $\delta_n^\text{AC}$, which influence the magnitude of the refractive index perturbations, are sometimes connected via the \textit{fringe visibility} $\nu$ such that $\delta_n^\text{AC}=\nu\delta_n^\text{DC}$, but we employ the slightly more general setting here. 
In this paper the base period length of the grating is denoted by $\Lambda$, and all deviations from this basic period are encoded in the function $\Phi$. 
The functions $\dndc,\dnac$ and $\Phi$ allow to describe complex grating structures including common FBG designs as those shown in Figure~\ref{fig:FBGtypes}.

The refractive index in the fibre core is then given by
\begin{equation} \label{eq:n_total}
	n(z) = \neff + \delta_n(z), 
	\qquad -
	\frac{L}{2} \leq z \leq \frac{L}{2}.
\end{equation}

\section{Forward Model} \label{sec:forward}

The forward model can be derived from Maxwell's equations combined with \emph{coupled mode theory}. 
The complete derivation is rather technical and does not provide much insight into the problem from a mathematical point of view. 
We therefore only sketch the most important steps. 
For an exhaustive derivation we refer to \cite{Kashyap1999}. 

From Maxwell's equations one derives the wave equation
\begin{equation} \label{eq:waveeq}
	\Delta\vec{E}
	=
	\epsilon_0\mu_0 \frac{\partial^2 \vec{E}}{\partial t^2}
	+
	\mu_0\frac{\partial^2 \vec{P}}{\partial t^2}+\nabla(\nabla\cdot \vec{E}),
\end{equation}
relating the electric field $\vec{E}$ and the polarization $\vec{P}$ inside the fibre. 
Here $\mu_0,\epsilon_0$ are vacuum permeability and permittivity, respectively. Using the weakly guiding approximation (assuming the electric field has only transversal components), one can drop the term $\nabla(\nabla\cdot \vec{E})$. 
The electric field and the polarization are further related via $\vec{P}=\epsilon_0(\epsilon_r-1)\vec{E}$, where $\epsilon_r$ the relative permittivity of the medium (here: the fibre core). 
In non-magnetic materials such as optical fibres, the permittivity is connected to the refractive index via $\epsilon_r=n^2$. 
In case of perturbations, we write
\begin{equation}
\vec{P}=\vec{P}_0+\vec{P}_{\delta_n}=\epsilon_0(\epsilon_r-1)\vec{E}+\epsilon_0\Delta\epsilon(z)\vec{E}
\end{equation}
and
\begin{equation}
	\epsilon_r+\Delta\epsilon(z)
	=
	(n+\delta_n(z))^2.
\end{equation}
Typically the term $\delta_n^2$ is neglected, as it is generally at least two orders of magnitude smaller than $n$. We thus have
\begin{equation}
	\Delta\epsilon(z)\approx 2n\delta_n(z),
\end{equation}
and the polarization can be written as
\begin{equation}
	\vec{P} = \epsilon_0(\epsilon_r-1)\vec{E} + 2\epsilon_0 n_0 \delta_n(z) \vec{E}.
\end{equation}
Inserting this into \eqref{eq:waveeq} yields, after some rearrangement,
\begin{equation} \label{eq:waveeq2}
	\Delta\vec{E} - \epsilon_0  \mu_0 \epsilon_r \frac{\partial^2 \vec{E}}{\partial t^2}
	=
	2\mu_0\epsilon_0n_0\delta_n(z)\frac{\partial^2 \vec{E}}{\partial t^2}.
\end{equation}
Next, one makes an ansatz for the electric field. For the single mode fibres we consider here, we set
\begin{equation}\label{eq:ansatz}
\vec{E}(x,y,z,t)=A_f(z)\exp(i\beta_fz+i\omega t)\vec{e}_{f}(x,y)+A_b(z)\exp(-i\beta_bz+i\omega t) \vec{e}_{b}(x,y).
\end{equation}
Here, $\vec{e}_{f}$ and $\vec{e}_{b}$ are the (complex valued) transversal electric fields of the forward ($f$) and backward ($b$) propagating waves which propagate along the fibre with amplitude $A_f(z)$ and $A_b(z)$, and with propagation constants $\beta_f$ and $\beta_b$ (notice the different signs of the propagation constants indicate reversed direction), respectively. We insert \eqref{eq:ansatz} into \eqref{eq:waveeq2}, which yields after some calculation the relation

\begin{align}\label{eq:we_ansatznserted}
2i\beta_f&\frac{\partial A_f(z)}{\partial z}\exp(i\beta_fz+i\omega t)\vec{e}_{f}-2i\beta_b\frac{\partial A_b(z)}{\partial z}\exp(-i\beta_bz+i\omega t)\vec{e}_{b}\nonumber\\&=-2\omega^2\mu_0\epsilon_0n_0\delta_n(z)\left(A_f(z)\exp(i\beta_fz+i\omega t)\vec{e}_{f}+ A_b(z)\exp(-i\beta_bz+i\omega t) \vec{e}_{b}\right).
\end{align}
In particular, we have employed the \textit{slowly varying envelope approximation}, i.e., $\frac{\partial^2 A_{\cdot}(z)}{\partial z^2} \ll \beta_{\cdot}\frac{\partial A_{\cdot}(z)}{\partial z}$, which allows to neglect second-order derivatives. Next we split \eqref{eq:we_ansatznserted} into two equations by multiplying with the complex conjugate transversal fields $\vec{e}_f^\ast$ and $\vec{e}_b^\ast$, and integrating over the fibre cross section ($x$-$y$-plane). We will only purse the first path further, the latter follows analogously. Note that in the absence of refractive index perturbations the transversal fields are orthonormal, $\int \vec{e}_{b} \vec{e}_{b}^\ast\, d(x,y)=1$, $\int \vec{e}_{b} \vec{e}_{f}^\ast\, d(x,y)=0$ and analogously for $b$ and $f$ exchanged. Due to conservation of energy (the frequency of the generated wave must not change), we can divide \eqref{eq:we_ansatznserted} by $\exp(i\omega t)$ and, all together, obtain
\begin{align}\label{eq:coupledA}
&\frac{2}{i\omega \epsilon_0n_0}\frac{\partial A_f(z)}{\partial z}\exp(i\beta_fz)\nonumber\\
&=\left(A_f(z)\exp(i\beta_fz)\int\delta_n(z)\vec{e}_{f}\vec{e}_{f}^\ast\, d(x,y)+ A_b(z)\exp(-i\beta_bz)\int \delta_n(z) \vec{e}_{b}\vec{e}_{f}^\ast\, d(x,y)\right),
\end{align}
where we insert the grating refractive index distribution \eqref{eq:grating}, in which we replace $\cos(\cdot)=\frac{1}{2}(\exp(i \cdot)+\exp(-i\cdot))=:\frac{1}{2}(\exp(i\cdot)+c.c.)$ for brevity. This yields
\begin{align}
&\frac{2}{i\omega \epsilon_0n_0}\frac{\partial A_f(z)}{\partial z}\exp(i\beta_fz)\label{eq:cAlong_1}\\
&=A_f(z)\exp(i\beta_fz)\int\left(\delta_n^{DC}+\frac{\delta_n^{AC}(z)}{2}\left(\exp\left(i\left(\frac{2\pi z}{\Lambda}+\Phi(z) \right)\right)+c.c.\right)\right)\vec{e}_{f}\vec{e}_{f}^\ast\, d(x,y) \label{eq:cAlong_2}\\
&+A_b(z)\exp(-i\beta_bz)\int\left(\delta_n^{DC}+\frac{\delta_n^{AC}(z)}{2}\left(\exp\left(i\left(\frac{2\pi z}{\Lambda}+\Phi(z) \right)\right)+c.c.\right)\right)\vec{e}_{b}\vec{e}_{f}^\ast\, d(x,y).\label{eq:cAlong_3}
\end{align}
From this formulation, we continue with the last crucial step, the so called phase-matching. The principle behind this is that modes can only exchange energy if they are in phase for a sufficiently long way. Comparing the phases in \eqref{eq:cAlong_1} and \eqref{eq:cAlong_2}, we see that $\frac{\partial A_f(z)}{\partial z}$ and the $\dndc$-component of $A_f(z)$ are always in phase $\beta_fz$, while this is never the case for $\frac{\partial A_f(z)}{\partial z}$ and the $\dnac$-component of $A_f(z)$ due to the modulation $\frac{2\pi z}{\Lambda}+\Phi(z)$. On the other hand, comparing the phases in \eqref{eq:cAlong_1} and \eqref{eq:cAlong_3}, we see that the $\dndc$-term of $A_b(z)$ is never in phase with $\frac{\partial A_f(z)}{\partial z}$, while this may (partially) be the case for the $\dnac$-term. To describe the latter properly, one introduces the \textit{phase-synchronous factor} $\beta_\phi:=\frac{2\pi}{\Lambda}+\beta_b$, such that a continuous transfer of energy occurs for $\beta_f=\beta_\phi$. We define the phase mismatch
\[
\Delta \beta:=\beta_f-\beta_\phi=\beta_f+\beta_b-\frac{2\pi}{\Lambda}.
\]
For the single mode fibres under consideration, the remaining integrals are easily evaluated and yield the coupling coefficients
\begin{equation}\label{eq:kappadc}
\kappa_{DC}(z):=n_0\omega\epsilon_0\int_F \delta_n^{DC}\vec{e}_{jt}\vec{e}_{jt}^\ast\,d(x,y)=\frac{2\pi}{\lambda}\dndc(z)
\end{equation}
and the $AC$-coupling coefficient
\begin{equation}\label{eq:kappaAC}
\kappa_{AC}(z):=n_0\omega\epsilon_0\int_F \frac{\delta_n^{AC}}{2}\vec{e}_{kt}\vec{e}_{jt}^\ast\,d(x,y)=\frac{\pi}{\lambda}\dnac(z).
\end{equation}
Now we can write Eqs. \eqref{eq:cAlong_1}-\eqref{eq:cAlong_3} in a more compact way as
\begin{equation}\label{eq:couplekappaA}
\frac{\partial A_f(z)}{\partial z}=i \frac{2\pi}{\lambda}\dndc(z) A_f(z)+i\frac{\pi}{\lambda}\dnac(z) A_b(z)\exp(-i(\Delta\beta z-\Phi(z)).
\end{equation}
Repeating the same steps for the second branch of \eqref{eq:we_ansatznserted}, we obtain
\begin{equation}\label{eq:couplekappaB}
\frac{\partial A_b(z)}{\partial z}=-i\frac{2\pi}{\lambda}\dndc(z)A_b(z)-i\frac{\pi}{\lambda}\dnac(z) A_f(z)\exp(i(\Delta\beta z-\Phi(z))).
\end{equation}
Finally, we substitute
\begin{align}
R(z,\lambda):&=A_f(z)e^{i(\Delta \beta(\lambda) z-\Phi(z))/2}\nonumber\\
S(z,\lambda):&=A_b(z)e^{-i(\Delta \beta(\lambda) z-\Phi(z))/2}\label{eq:RS_pair}
\end{align}
(where $R$ is short for reference, or the light going into the FBG, and $S$ is the signal, or measured, reflected light), and adopt the common notation
\begin{equation}\label{eq:sigmahat}
\hat \sigma(z,\lambda):=\delta(\lambda)+\sigma(z,\lambda)-\frac{1}{2}\frac{\partial \Phi}{\partial z}
\end{equation}
with the convention (note: for single-mode fibres $|\beta_f|=|\beta_b|=:\beta$)
\begin{equation} \label{eq:delta}
	\delta(\lambda)
	=
	2\Delta\beta(\lambda) 
	=
	\beta-\frac{\pi}{\Lambda}=\beta-\beta_D
	=
	2\pi n_\text{eff}\left(\frac{1}{\lambda}-\frac{1}{\lambda_D}\right),
\end{equation}
$\lambda_D$ from \eqref{eq:gratingcondition} and (for a single-mode fibre) 
\[
	\sigma(z,\lambda)
	=
	\kappa^\text{DC}(z,\lambda)
	=
	\frac{2\pi}{\lambda}\delta_n^\text{DC}(z).
\]
This finally yields the well-known system of coupled ordinary differential equations 

\begin{align}\label{eq:coupledDE_erdogan}
	\frac{d R(z,\lambda)}{dz}-i\hat\sigma(z,\lambda) R(z,\lambda)
	&=
	i\kappa_\text{AC}(z,\lambda) S(z,\lambda)\nonumber\\
	\frac{d S(z,\lambda)}{dz}+i\hat\sigma(z,\lambda) S(z,\lambda)
	&=
	-i\kappa_\text{AC}(z,\lambda) R(z,\lambda)
\end{align}
for $z$ in the grating zone, i.e., $-\frac{L}{2}\leq z\leq\frac{L}{2}$, of an FBG of length $L>0$. 
To find a solution, one typically imposes the boundary conditions $R(-\frac{L}{2},\lambda)=1$ (light enters the FBG with full intensity) and $S(\frac{L}{2},\lambda)=0$ (no light is reflected beyond the FBG).

Finally, the measured data corresponds to the (relative) intensity of the reflected light,
\begin{equation} \label{eq:data}
	r(\lambda)
	=
	\left|\frac{S(-\frac{L}{2},\lambda)}{R(-\frac{L}{2},\lambda)} \right|^2.
\end{equation}
Note that the coupled system \eqref{eq:coupledDE_erdogan} has complex-valued coefficients, but the measured data consists only of intensities. 
The phase information is lost. 
Therefore, the task of recovering the functions $\dndc,\dnac$, and $\Phi$, which are hidden in the coefficients $\hat\sigma$ and $\kappa_{AC}$, from the data $r(\lambda)$, is a \emph{phase retrieval problem} with a nonlinear forward operator.

Since one is only interested in the reflected spectrum, the system of equations \eqref{eq:coupledDE_erdogan} can be rewritten as a Ricatti-type differential equation by making the substitution $\rho(z,\lambda):=\frac{S(z,\lambda)}{R(z,\lambda)}$. 
This yields
\begin{equation}\label{eq:ricatti_ode}
	\frac{d\varrho(z,\lambda)}{d z}
	=
	-i\kappa_{AC}(z,\lambda)-2 i\hat\sigma(z,\lambda)\varrho(z)-i\kappa_{AC}(z,\lambda)\varrho(z)^2,
	\qquad \varrho\left(\frac{L}{2}\right)=0,
\end{equation}
and $r(\lambda)=|\rho(\lambda)|^2$. 
In case of a homogeneous grating, i.e., $\dndc(z)\equiv \mathrm{const}$, $\dnac(z)\equiv \mathrm{const}$, and $\Phi(z)\equiv 0$ such that $\hat\sigma(z,\lambda)$ and $\kappa_{AC}(z,\lambda)$ are also constant for fixed $\lambda$, then \eqref{eq:coupledDE_erdogan} (and \eqref{eq:ricatti_ode}, respectively) have an analytic solution. Defining
\begin{equation}\label{eq:gamma}
\gamma(\lambda):=\sqrt{\kappa(\lambda)^2-\hat\sigma(\lambda)^2}
\end{equation}
one can show that
\begin{equation}\label{eq:analytic_solution}
r(\lambda)=\frac{\sinh^2(\gamma(\lambda) L)}{\cosh^2(\gamma(\lambda) L)-\frac{\hat\sigma(\lambda)^2}{\kappa(\lambda)^2}}.
\end{equation} 
In practice, integrating \eqref{eq:coupledDE_erdogan} or \eqref{eq:ricatti_ode} numerically is rather time consuming. Instead, one commonly uses the \emph{transfer matrix approach} \cite{Erdogan1997}. 
The idea behind this approach is to divide the Bragg grating of length $L$ in $M$ subelements of uniform length $\Delta z$ that are sufficiently long ($\Delta z \gg \Lambda$). 
For each subelement $j$, $j=1,\dots,M$, one has
\begin{equation}
\begin{pmatrix}
R_j\\S_j
\end{pmatrix}=T_j\begin{pmatrix}
R_{j-1}\\ S_{j-1}
\end{pmatrix}
\end{equation}
where $(R_j,S_j)^T$ are the amplitudes of ingoing and reflected spectrum before the subelement, $(R_{j-1},S_{j-1})^T)$ the amplitudes after the subelement, and the transfer matrix $T_j$ is given by
\begin{equation}
T_j=\begin{pmatrix}
\cosh(\gamma_j \Delta z)-i\frac{\hat\sigma_j}{\gamma_j}\sinh(\gamma_j \Delta z) & -i\frac{\kappa_j}{\gamma_j}\sinh(\gamma_j \Delta z)\\ \frac{\kappa_j}{\gamma_j}\sinh(\gamma_j \Delta z) & \cosh(\gamma_j \Delta z)+i\frac{\hat\sigma_j}{\gamma_j}\sinh(\gamma_j \Delta z)
\end{pmatrix}.
\end{equation}
Above, $\gamma_j:=\sqrt{\kappa_{AC}(z_j,\lambda)^2-\hat\sigma(z_j,\lambda)^2}$, and $\hat\sigma_j,\kappa_j$ are the coefficients \eqref{eq:sigmahat} and \eqref{eq:kappaAC} averaged over the subelement $j$. Numerically, $\gamma_j,\hat\sigma_j,\kappa_j$ are evaluated as $\gamma(z_j),\hat\sigma(z_j),\kappa(z_j)$ where $z_j$ is the central coordinate of the subelement $j$.

One then links all subelements by multiplying their transfer matrices such that with
\begin{equation}\label{eq:TM_assembly}
T:=\Pi_{j=1}^M T_j
\end{equation}
one has
\begin{equation}\label{eq:TM_equation}
\begin{pmatrix}
R(-\frac{L}{2})\\S(-\frac{L}{2})
\end{pmatrix}=T\begin{pmatrix}
R(\frac{L}{2})\\S(\frac{L}{2})
\end{pmatrix}.
\end{equation}
Together with the boundary conditions ($R(-\frac{L}{2},\lambda)=1$ and $S(\frac{L}{2},\lambda)=0$) one finally obtains
\begin{equation}\label{eq:reflectedspecTM}
r(\lambda)=|\rho(\lambda)|^2=\left|\frac{T_{21}(\lambda)}{T_{11}(\lambda)} \right|^2
\end{equation}
for the reflection intensity.

To sum up, the forward operator 
\begin{equation}\label{eq:forward_op}
F:(L_2(\R)\times L_2(\R)\times L_2(\R))\rightarrow L_2(\R),\quad(\dndc(z),\dnac(z),\Phi(z))\mapsto r(\lambda)
\end{equation}
 maps three (potentially) unknown functions to the spectrum. To discretize the problem, we use discretization of the unknowns implied by the transfer matrix formulation. 
The spectrum is discretized by pointwise evaluation, which in practice is determined by the measurement device. 
Note that the rather coarse discretization here reduces the degrees of freedom significantly, and acts as a first important regularization. 
Even more, fine discretization of the unknown functions over the entire grating length is infeasible. 
In order to fully resolve $n(z)$ as in \eqref{eq:n_total}, a high spatial resolution is needed to avoid aliasing effects.
For example, a typical period length $\Lambda=532 \mathrm{nm}$ in \eqref{eq:dnreal} would require roughly 4000 discretization points per millimetre FBG-length to satisfy the Nyquist sampling criterion, if one wanted to discretize \eqref{eq:dnreal} in a straightforward way. 
As a comparison, our experimental setup yields about 10 significant points $r(\lambda_i)$ in the spectrum.

\section{Uniqueness} \label{sec:unique}

An important question for the inversion of an FBG-spectrum is its uniqueness: 
Is a spectrum uniquely determined by the three functions $\dndc,\dnac,\frac{\partial\Phi}{\partial z}$, or can multiple configurations produce the same data? 
We have only found one partial answer in the literature. 
In \cite{GillPeters2004} it was mentioned that the period distribution $\Lambda(z)$ (corresponding in our notation to $\frac{\partial\Phi}{\partial z})$ is reversible without affecting the data, i.e., one can not discern between $\frac{\partial\Phi}{\partial z}(z)$ and $\frac{\partial\Phi}{\partial z}(-z)$. This is true in their setting, but not in general. Except for this reversal, the authors claim uniqueness with respect to the distribution of $\Phi(z)$ over the FBG, and attribute it to the fact that the matrix multiplication in the transfer matrix formulation in \eqref{eq:TM_assembly}, \eqref{eq:reflectedspecTM} is not commutative. However, this only guarantees that the transfer matrix does not introduce additional non-uniqueness. The actual source of non-uniqueness lies in the physical model. What holds is the main argument of \cite{GillPeters2004} that the FBG can not distinguish between light propagating in forward and backward direction. The following theorem explains this mathematically and gives an extended result on the uniqueness of the FBG spectra. We stress that we do not claim completeness of the list. 

\begin{theorem} \label{thm:unique}
Let $r(\lambda)=F(\dndc,\dnac,\Phi)$ be a reflected FBG spectrum.
Then the following ambiguities hold:
\begin{enumerate}[(a)] \setlength\itemsep{.3em}
\item $F(\dndc,\dnac,\Phi)=F(\dndc,\dnac,\Phi+c)$ for all $c\in \R$,
\item $F(\dndc,\dnac,\Phi)=F(\dndc,-\dnac,\Phi)$,
\item $F(\dndc,\dnac(z),\Phi)=F(\dndc,\dnac(-z),\Phi)$ if $\dndc(z)=\dndc(-z)$, $\Phi(z)=-\Phi(-z)$,
\item $F(\dndc(z),\dnac,\Phi(z))=F(\dndc(-z),\dnac,-\Phi(-z))$ if $\dnac(z)=\dnac(-z)$.
\end{enumerate}
\end{theorem}

\begin{proof}
Since the coupled system \eqref{eq:coupledDE_erdogan} only uses $\frac{\partial \Phi}{\partial z}$, any constant in $\Phi$ disappears, hence a) holds.
Because of \eqref{eq:kappaAC}, $\dnac$ and $\kappa_{AC}$ have the same sign, which is lost due to taking the absolute value in \eqref{eq:data}. This is easily seen from \eqref{eq:ricatti_ode}. Multiplying the equation by $-1$, we can write
\[
	\frac{d\left(-\varrho(z,\lambda)\right)}{d z}
	=
	-i(-\kappa_{AC}(z,\lambda))
	-2 i\hat\sigma(z,\lambda)(-\varrho(z,\lambda))
	-i(-\kappa_{AC}(z,\lambda))\varrho(z)^2,
\]
which yields the identical ODE for $-\rho$. 
The sign, however, is lost in the intensity \eqref{eq:reflectedspecTM}. 
Note that such a trivial ambiguity is common for nonlinear phase retrieval problems, see, e.g., \cite{autoconv2014}.

For c) and d), we use the reversibility of the FBG as noted in \cite{GillPeters2004}. Take the coupled system \eqref{eq:coupledDE_erdogan}, exchange the roles of $R$ and $S$, and set $t=-z$, i.e., $\tilde R(t,\lambda):=S(-z,\lambda)$ and $\tilde S(t,\lambda)=R(-z(\lambda)$. Then the system can be written as 
\begin{align}\label{eq:coupledDE_transformed}
\frac{d \tilde R(t,\lambda)}{dt}-i\hat\sigma(-t,\lambda) \tilde R(t,\lambda)&=i\kappa_{AC}(-t,\lambda) \tilde S(t,\lambda)\nonumber\\
\frac{d \tilde S(t,\lambda)}{dt}+i\hat\sigma(-t,\lambda) \tilde S(t,\lambda)&=-i\kappa_{AC}(-t,\lambda) \tilde R(t,\lambda),
\end{align}
which coincides with \eqref{eq:coupledDE_erdogan} if $\hat \sigma(-t,\lambda)=\hat\sigma(t,\lambda)$ and $\kappa_{AC}(-t,\lambda)=\kappa_{AC}(t,\lambda)$. Because only the derivative of $\Phi$ enters $\hat \sigma$, this yields the requirement $\frac{\partial\Phi(t)}{\partial t}=\frac{\partial\Phi(-t)}{\partial t}$, which holds if $\Phi(t)=-\Phi(-t)$.
\end{proof}

It is also worth mentioning that the non-uniqueness for FBGs is already inherent in the Bragg-wavelength \eqref{eq:gratingcondition}, $\lambda_B=2\neff\Lambda$. If neither $\neff$ nor $\Lambda$ are known, there are infinitely many pairs $(\neff,\Lambda)$ that produce the same $\lambda_B$. Related to this, we deem it important to report a potential ``pseudo non-uniqueness'' of the spectra with respect to practical measurements. Namely, despite being different, some configurations of the unknowns may yield spectra that are only distinguishable in small differences that may easily get lost during the measuring or (forward) simulation process. To exemplify this, we simulate two spectra with two distinct peaks at wavelengths $\lambda_1=1550$nm and $\lambda_2=1560$nm. All parameters are identical, except that in one spectrum we set 
\[
\dndc(z)=\begin{cases}0&z<0\\ \neff\left(\frac{1560nm}{1550nm}-1\right)&z\geq0\end{cases} \mbox{ and }\Phi(z)=0,
\]
whereas for the other spectrum we set 
\[
\dndc(z)=0 \mbox{ and } \Phi(z)=\begin{cases}0&z<0\\ 2\pi z\left(\frac{2\neff}{1560nm}-\frac{2n_{eff}}{1550nm} \right)  & z\geq 0\end{cases}.
\]
As shown in Figure \ref{fig:pseudounique}, this yields spectra that visibly only differ in the low intensity regions away from the peaks. The maximal difference is smaller than $10^{-2}$, which as discussed in the experiments in Section \ref{sec:experiments}, is already within the noise level.

\begin{figure}
\includegraphics[width=0.49\linewidth]{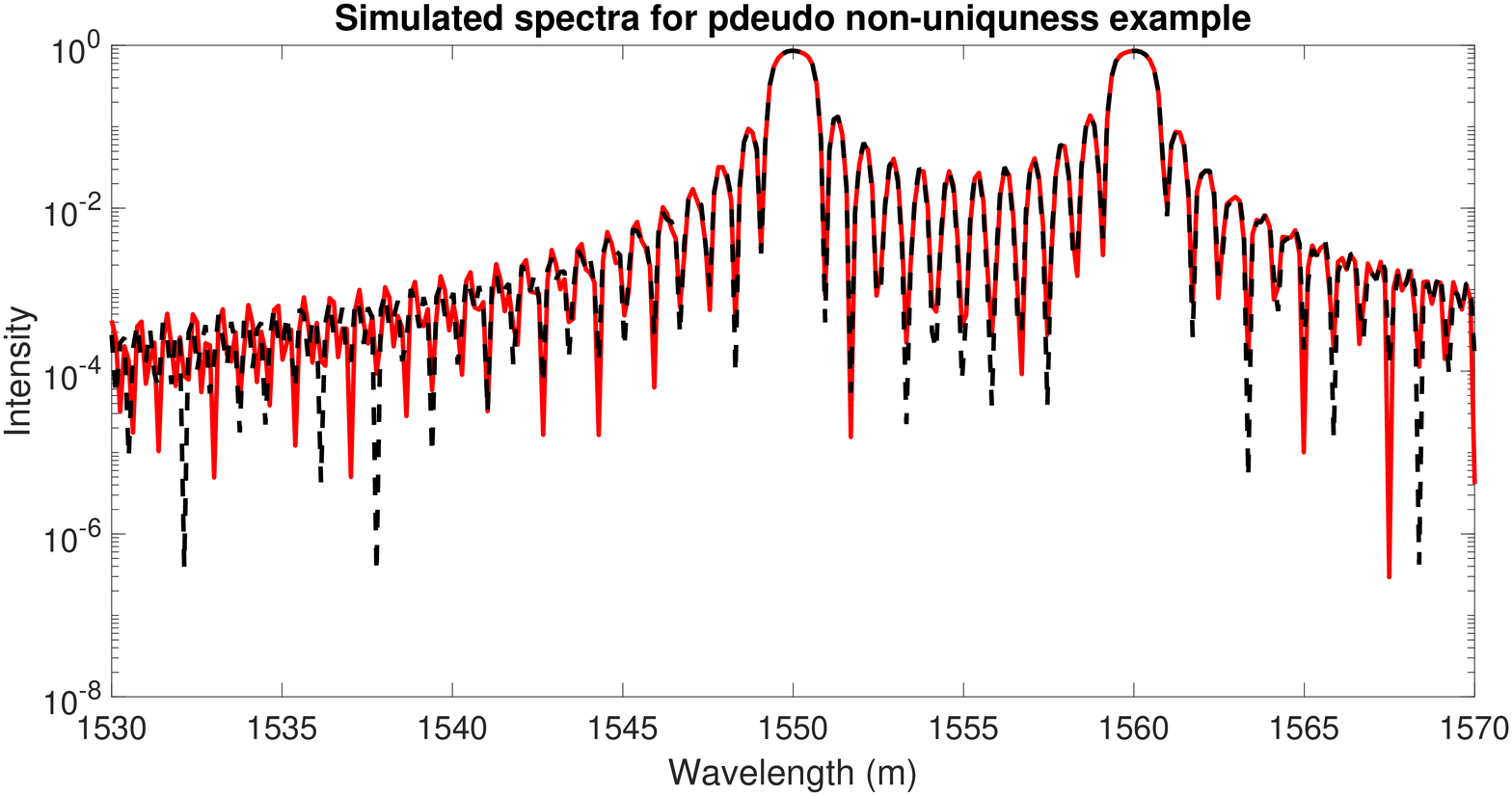}\includegraphics[width=0.49\linewidth]{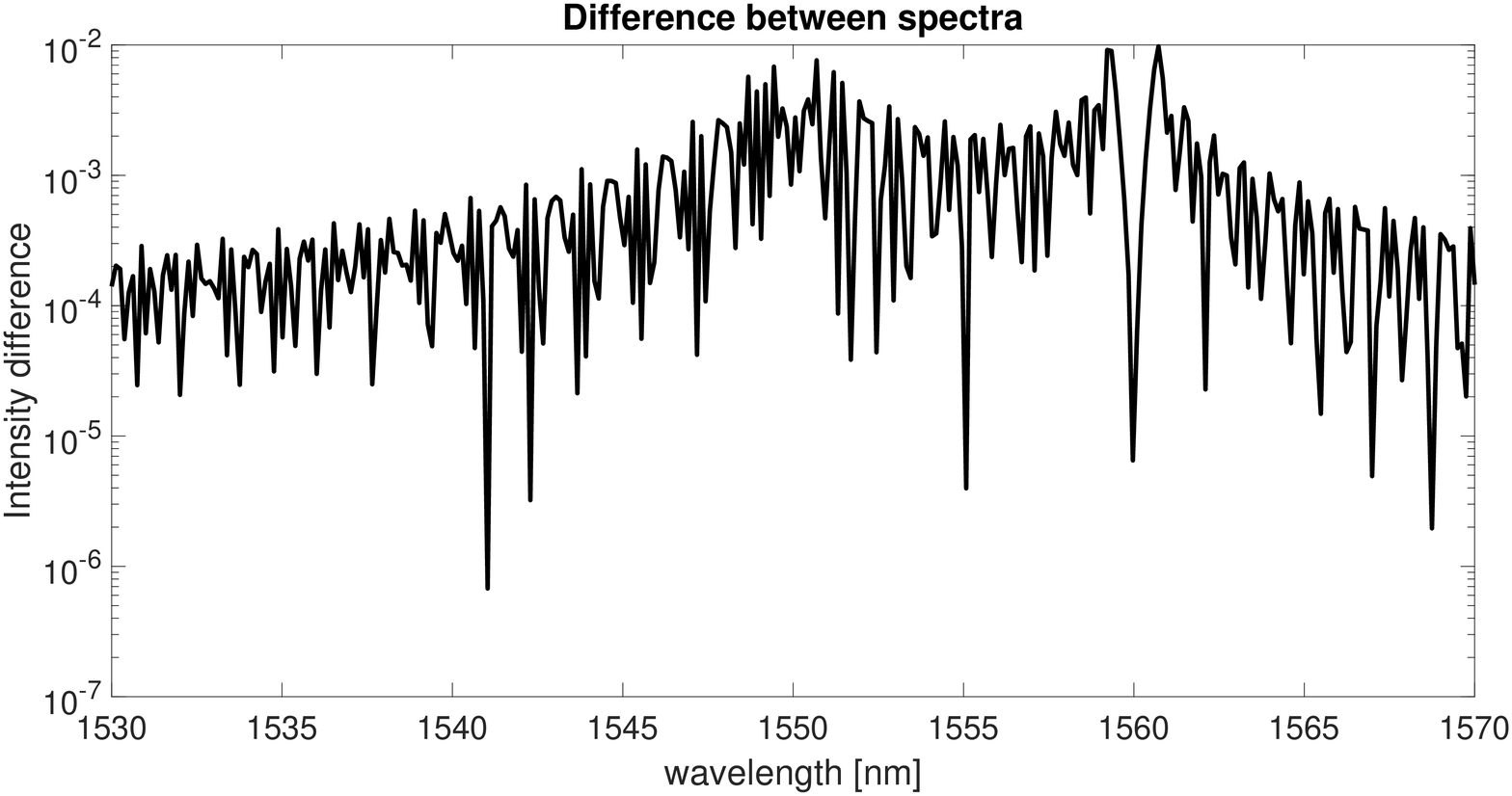}\caption{Simulated spectra (left) and their difference (right) from the pseudo non-uniqueness example. While the spectra are not identical, the peak position is, and all differences are in the small side lobes away from the peak. Only highly accurate measurements could distinguish between the spectra.}\label{fig:pseudounique}
\end{figure}

We have shown that FBG spectra are, in general, not uniquely determined and that they may appear identical despite varying slightly. 
Therefore, when conducting experiments, it is advisable to employ as much asymmetry as possible in the positioning or the expected reaction of the FBG, as well as use precise measuring devices.

\section{Regularization} \label{sec:reg}

In this section we present a regularization approach for several configurations of the potential unknowns. 
As explained in Section \ref{sec:forward}, the full problem consists of identifying three functions ($\dnac,\dndc,\Phi$) from the measured spectrum $r$. However, it is not always necessary to consider this general setting. 
Therefore, we begin with some special cases.

\subsection{Homogeneous FBG} \label{ssec:hom}

The most common FBG design is a \emph{homogeneous grating} characterized by
\[
   \dndc(z) \equiv \mathrm{const}, \quad
   \dnac(z) \equiv \mathrm{const}, \text{ and } 
   \Phi(z) \equiv 0. 
\]
Even then, in the authors' experience, the manufacturers of the FBG are not always able to provide precise values for $\dnac$ or $\dndc$. Therefore, in a first step, which can also be seen as a calibration for later test cases, we aim to recover the two parameters $\dndc,\dnac\in\R$ from the spectrum of a homogeneous FBG. An experimental validation is carried out in Section~\ref{sec:experiments}.

To find the unknown parameters, we calculate
\begin{equation}
[\dndc^*,\dnac^*]=\argmin_{\dndc,\dnac\in\R} \|F[\dndc,\dnac](\lambda)-\rmeas(\lambda)\|_2^2.
\end{equation}
The minimizer is calculated in MATLAB by way of the function \texttt{lsqnonlin} using the Levenberg-Marquardt algorithm. 
Due to the low dimensionality of the problem, no regularization is necessary. 
However, as usual for highly nonlinear problems, the initial guess may influence the quality of the reconstruction. 
Due to the simplicity of the setting it is fairly easy to find a good initial guess by hand, since the only two free parameters have distinct effects: $\dndc$ moves the spectrum along the $\lambda$-axis and $\dnac$ controls the intensity (height of the peak). An example of such a recovery is shown in Figure \ref{fig:dndcdnac_easy}. 
While it is not too surprising that we obtain perfect recovery of the parameters given the simplistic setting, we would like to highlight that this is an important calibration step for the experiments in practice. 
Specifically, the calculated value of $\dndc$ can be used to correct the position of the peak and thus potentially eliminates a source of error in the strain calculation via \eqref{eq:strain}.

\begin{figure}\centering
\includegraphics[width=0.7\linewidth]{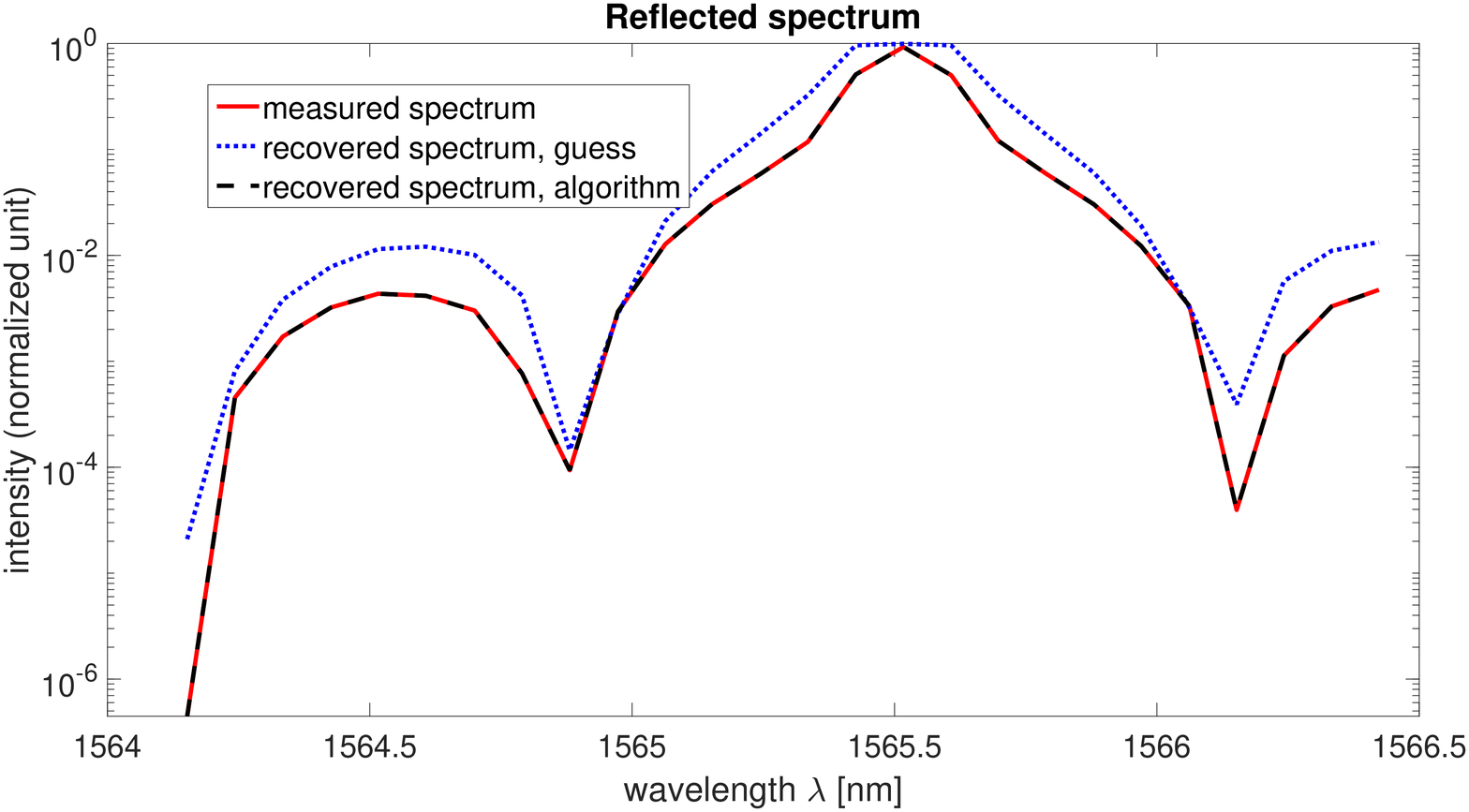}
\caption{Simulated and recovered spectrum for the retrieval of constant $\dndc$ and cosntant $\dnac$ as in Section \ref{ssec:hom}. For the initial guesses $\dnac=0.0002$ and $\dndc=0$, we obtained a perfect reconstruction of the true values $\dnac=0.0001234$ and $\dndc=-0.0000002$.}
\label{fig:dndcdnac_easy}
\end{figure}

\subsection{Tapered and Apodised FBGs} \label{ssec:apodized}

The need for regularization arises as soon as an inhomogeneous refractive index distribution is to be found, i.e., a combination of $\dndc,\dnac$, or $\Phi$ is to be recovered as functions resolved over the FBG length. In this section we consider $\dndc$ and $\dnac$ to be unknown, which can be interpreted as the determination of apodized and/or tapered FBGs (cf.\ Figure~\ref{fig:FBGtypes}). 

In our example we consider an apodized and tapered FBG of length $L=2\mathrm{mm}$ with $\Phi\equiv 0$ and
\begin{align*}
\dndc^\dag(z)&=\begin{cases}0.0002 & -0.001\leq z\leq 0\\ 0.1z+0.0002 & 0<z\leq 0.001\end{cases},\\
\dnac^\dag(z)&=\begin{cases} 0.0005 & -0.001\leq z \leq -0.0005\\ -0.2z+0.0004 & -0.0005<z\leq 0.0005 \\ 0.0003 & 0.0005<z\leq 0.001\end{cases}
\end{align*}

For the discretization we use the transfer matrix approach with $M=16$ subelements, corresponding to a resolution of $125\mathrm{\mu m}$. The data consists of $120$ measured points with a wavelength spacing of $0.16\mathrm{nm}$, which is of comparable size to the real data. 
To find the minimizer we again employ the \texttt{lsqnonlin}-function of MATLAB to calculate
\begin{align*}
&[\dndc^*(z),\dnac^*(z)]\\ &=\argmin_{\dndc,\dnac\in\R^{16}} \|F[\dndc(z),\dnac(z)](\lambda)-\rmeas(\lambda)\|_2^2+\alpha \|\dndc(z)\|_{H_1}^2+\beta\|\dnac(z)\|_{H^1}^2,
\end{align*}
i.e., we penalize the $H^1$-norm
\[
\| x\|_{H_1}^2=\|x\|_2^2+\|x^\prime\|_2^2
\]
to stabilize the recovery, and we now have two regularization parameters $\alpha$ and $\beta$. 
Among several tested penalty functionals, this gave the most reliable results. 
An appropriate choice of $\alpha$ and $\beta$ then yields a good approximation to the correct functions $\dndc$ and $\dnac$, as can be seen in Figure~\ref{fig:dndcdnac}. 
Hence, we have demonstrated that it is possible to extract information from the FBG that is far more detailed than the usual peak-evaluation. A crucial step and a centrepiece of future research is the automatization of the search of a suitable pair of regularization parameters $[\alpha,\beta]$. Currently we pick the regularization parameters by hand knowing the exact solutions, which is clearly not ideal. However, only few methods are documented in the literature for the simultaneous estimation of two regularization parameters, and none of them seem to work in our situation. One approach is a generalization of the L-curve method. In the classical setting with one regularization parameter, the L-curve was introduced in \cite{Hansen1993} and has become a popular heuristic parameter sekection rule. 
It features a $\log$-$\log$-plot of residual versus the corresponding penalty value, which often attains roughly the shape of an L. 
It is then advised to choose the regularized solution corresponding to the corner of the L. 
The principle has been extended to two-parameter regularization in \cite{BelgeEtAl2002} and recently successfully applied in \cite{HofHofPich}. 
In this case one plots the triples of residual and the two penalty values, each on a logarithmic scale. As generalization of the corner of the L, one then picks the regularized solution corresponding to the point on the ``L-surface" with highest curvature, typically to be found close to the origin. 
Unfortunately, this method appears to be unsuitable for the recovery of the FBG refractive index distribution, see Figure~\ref{fig:Lsurf}. 
While in principle every cross-section of the L-surface is L-shaped and thus the points of high curvature are easily found, a comparison with the plot of the total reconstruction error $e:=\sqrt{\|\dndc^\dag(z)-\dndc^*(z)\|_2^2+\|\dnac^\dag(z)-\dnac^*(z)\|_2^2}$ (see also Figure \ref{fig:Lsurf}) reveals that the best approximate solutions are found somewhere in the large, flat plane corresponding to the smaller regularization parameters. 
This plane, however, also contains the highly oscillating, underregularized approximations, as well as some approximations where the regularization parameters are already too large. 
The residual carries no meaningful information about the reconstruction error, so that other methods have to be developed. 
This observation also renders discrepancy-based parameter selection rules infeasible. 
Since our experimental goal is to calculate strain within a body, our plan is to link the FBG inversion with a finite element model of the body, and to find suitable regularization parameters through linking the two approaches.

\begin{figure}
\includegraphics[width=0.49\linewidth]{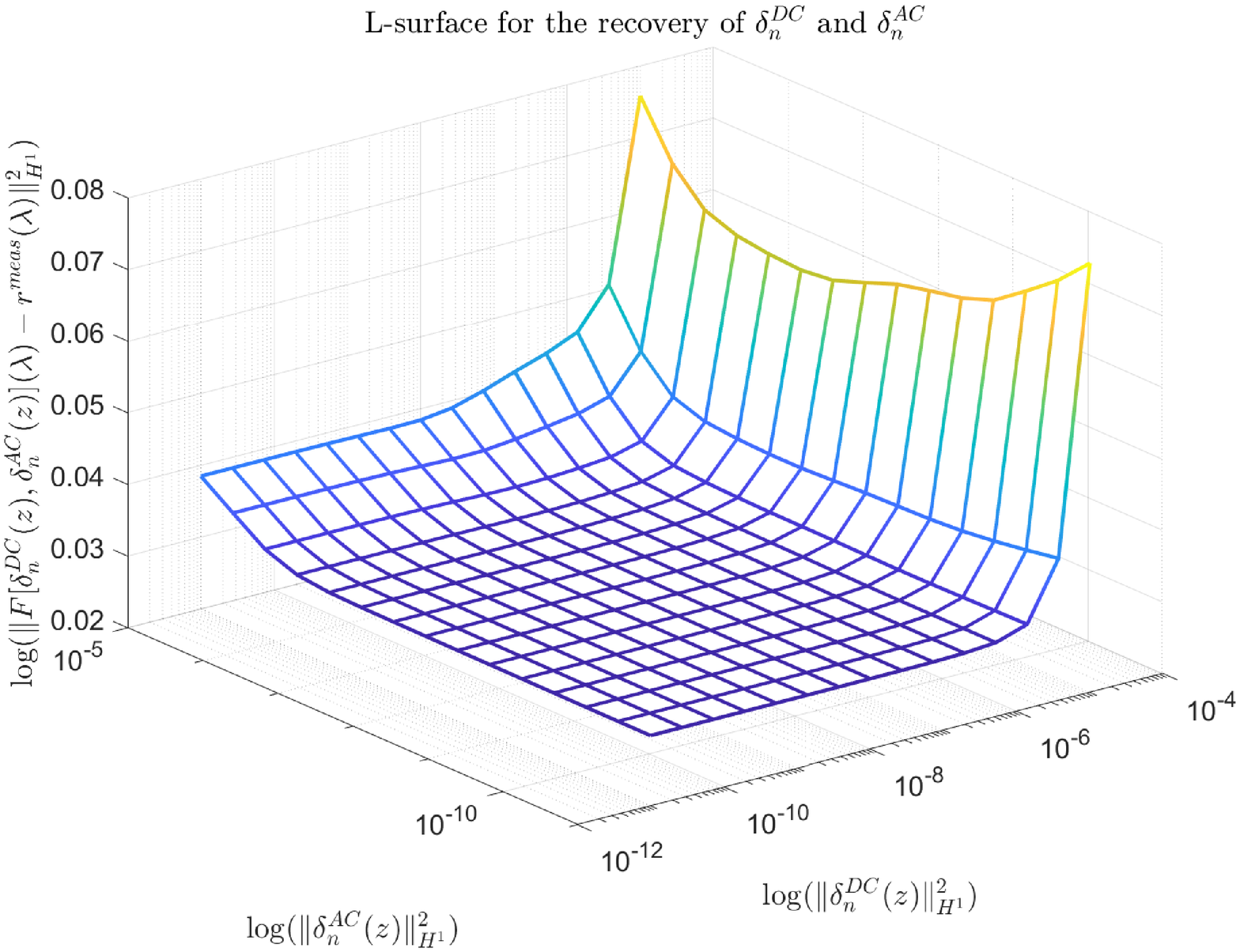}
\includegraphics[width=0.49\linewidth]{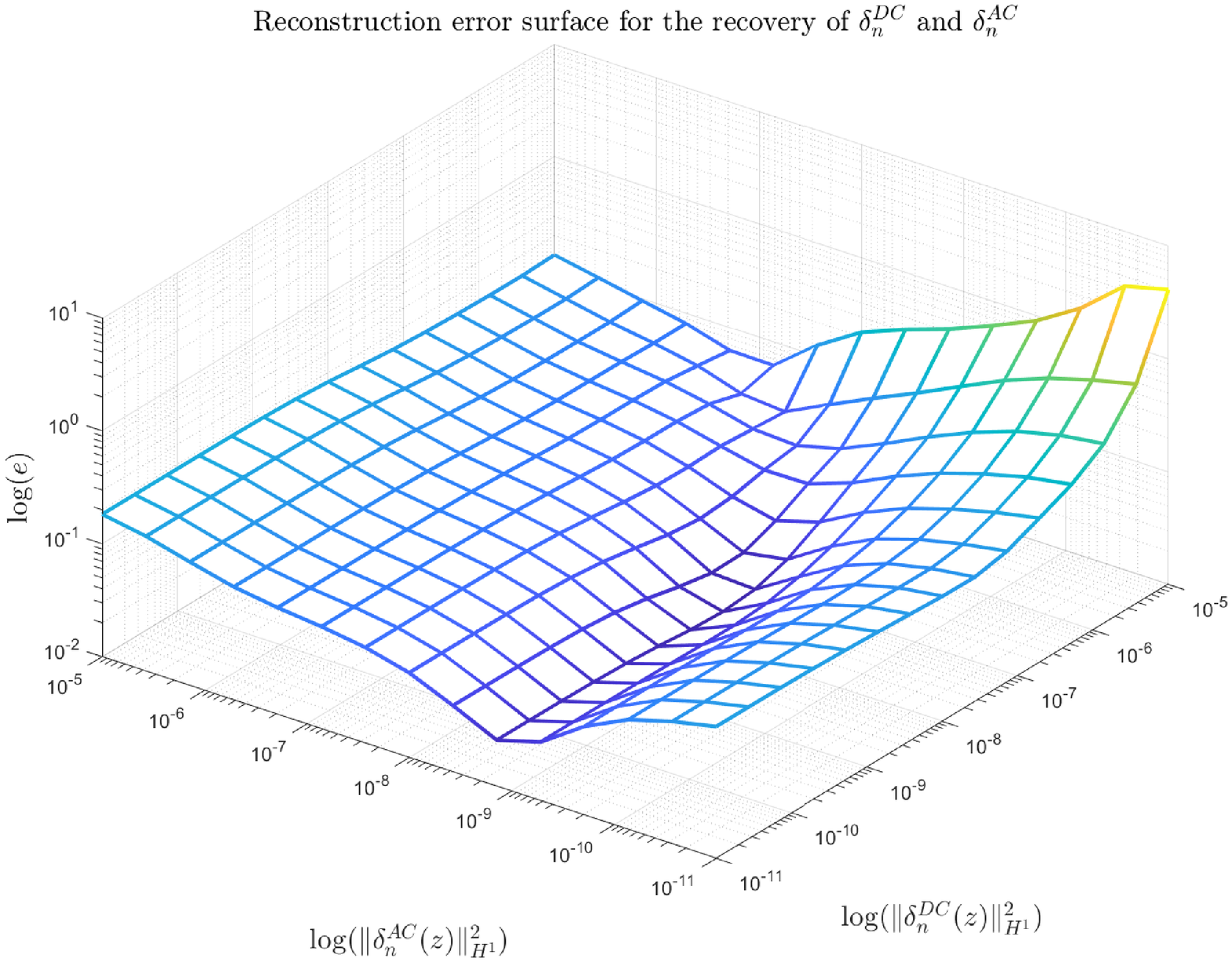}
\caption{Experiment of Section \ref{ssec:apodized}, L-surface (left) and reconstruction error (right). X- and y-axis correspond to the regularization parameters. The L-surface chooses the point of highest curvature. A comparison with the recosntruction shows that the actual minimum is far away from the corner of the L, instead the minimum lies in a flat plateau of the L-surface. Hence, this heuristic fails.}\label{fig:Lsurf}
\end{figure}

\begin{figure}
\includegraphics[width=0.32\linewidth]{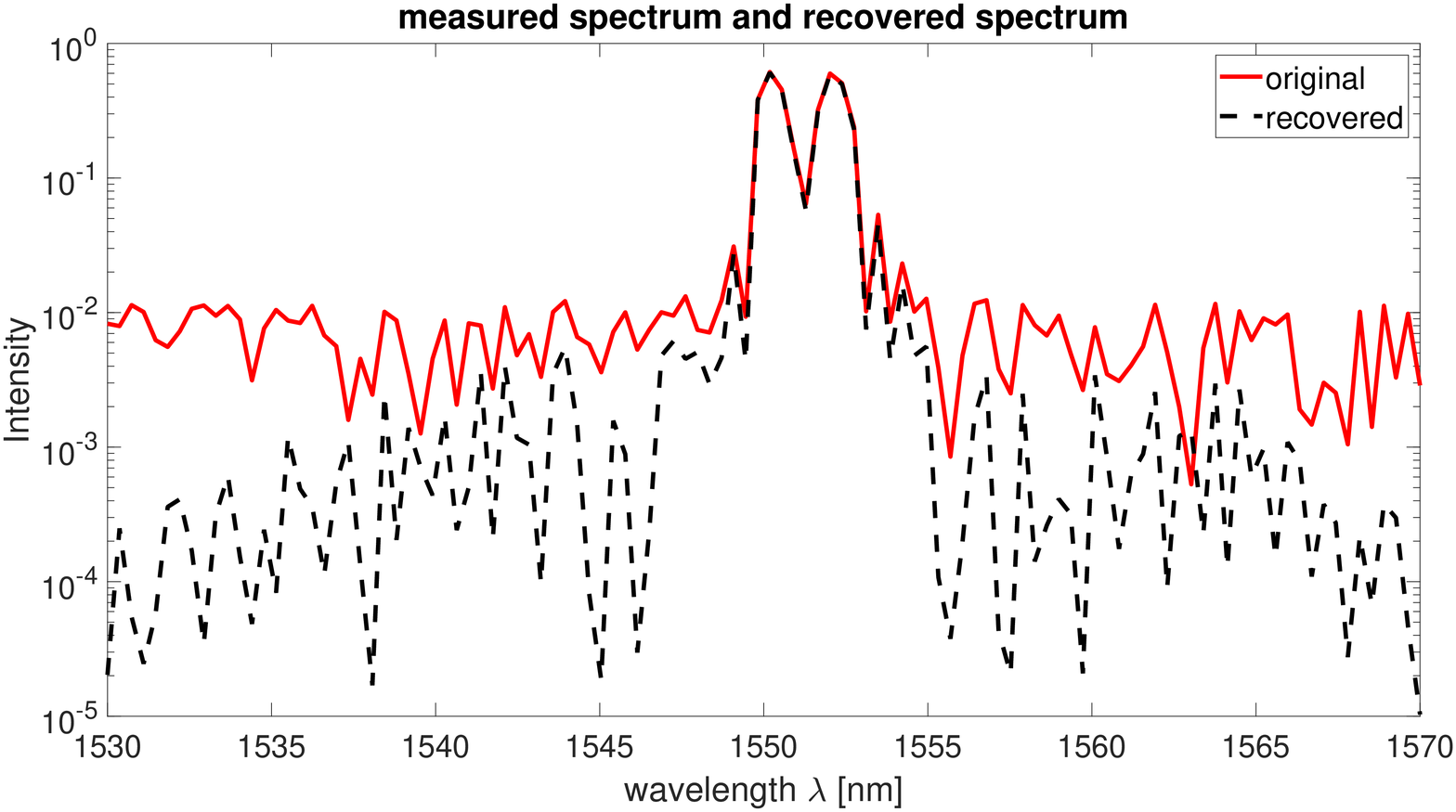}\includegraphics[width=0.32\linewidth]{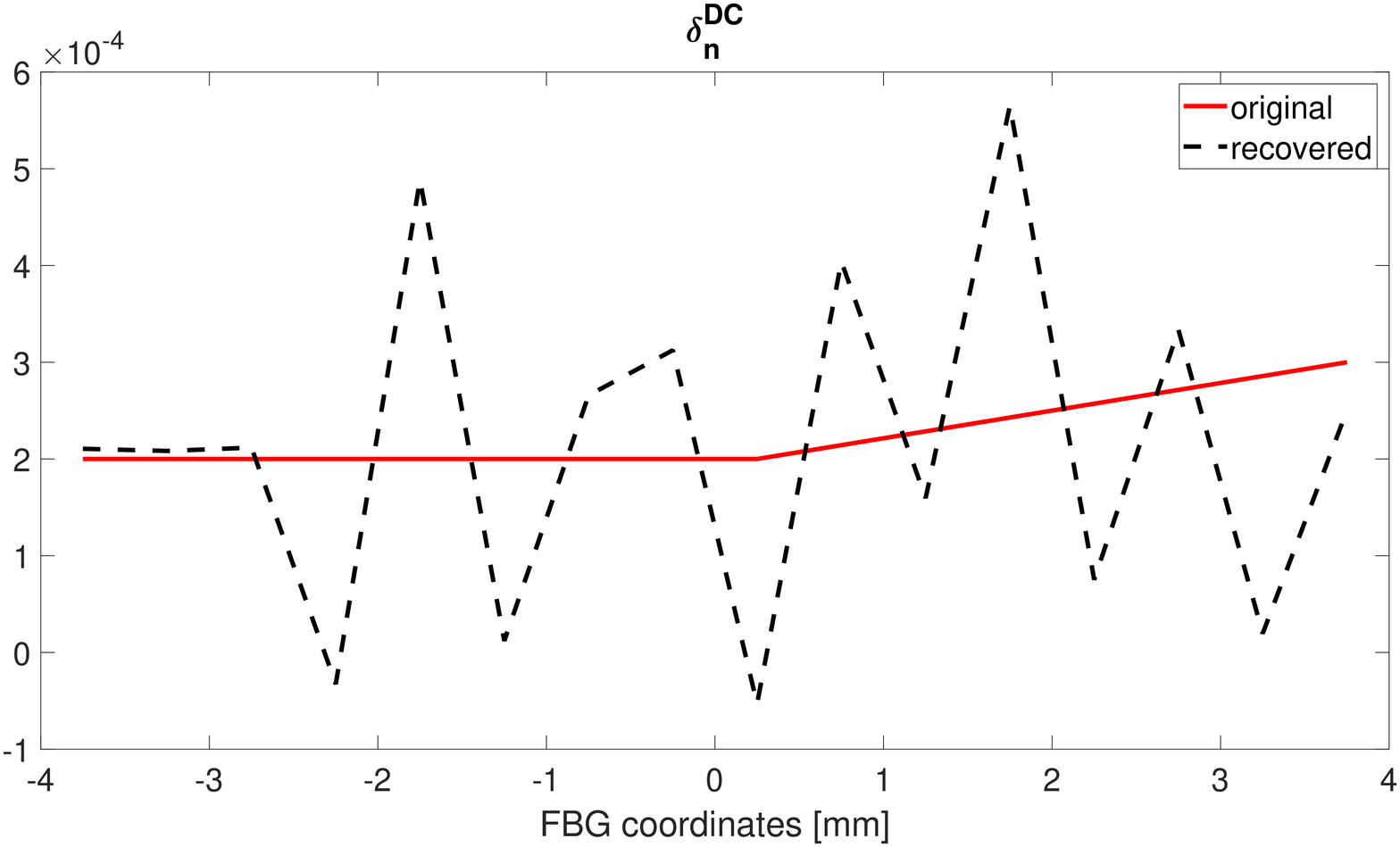}\includegraphics[width=0.32\linewidth]{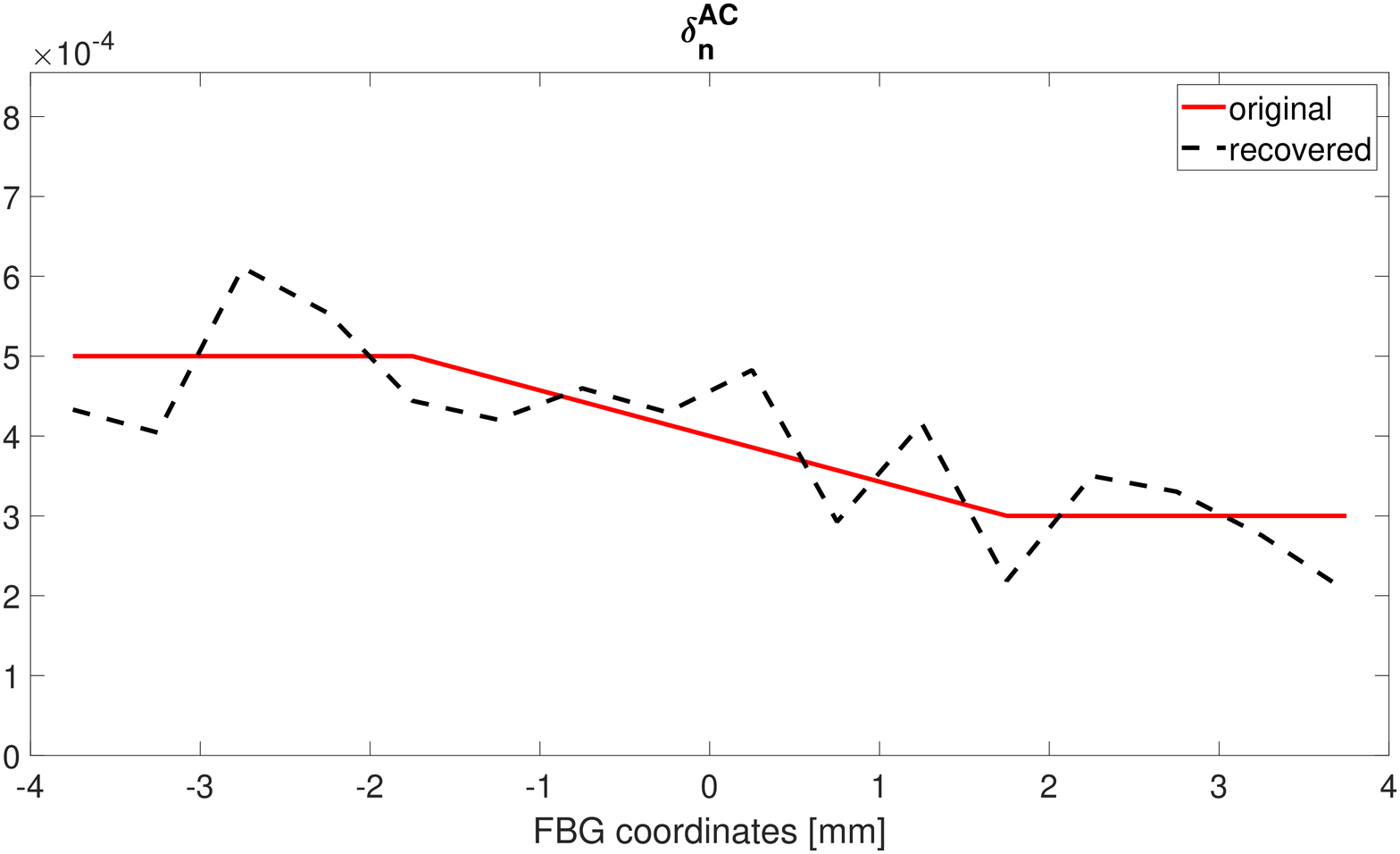}\\
\includegraphics[width=0.32\linewidth]{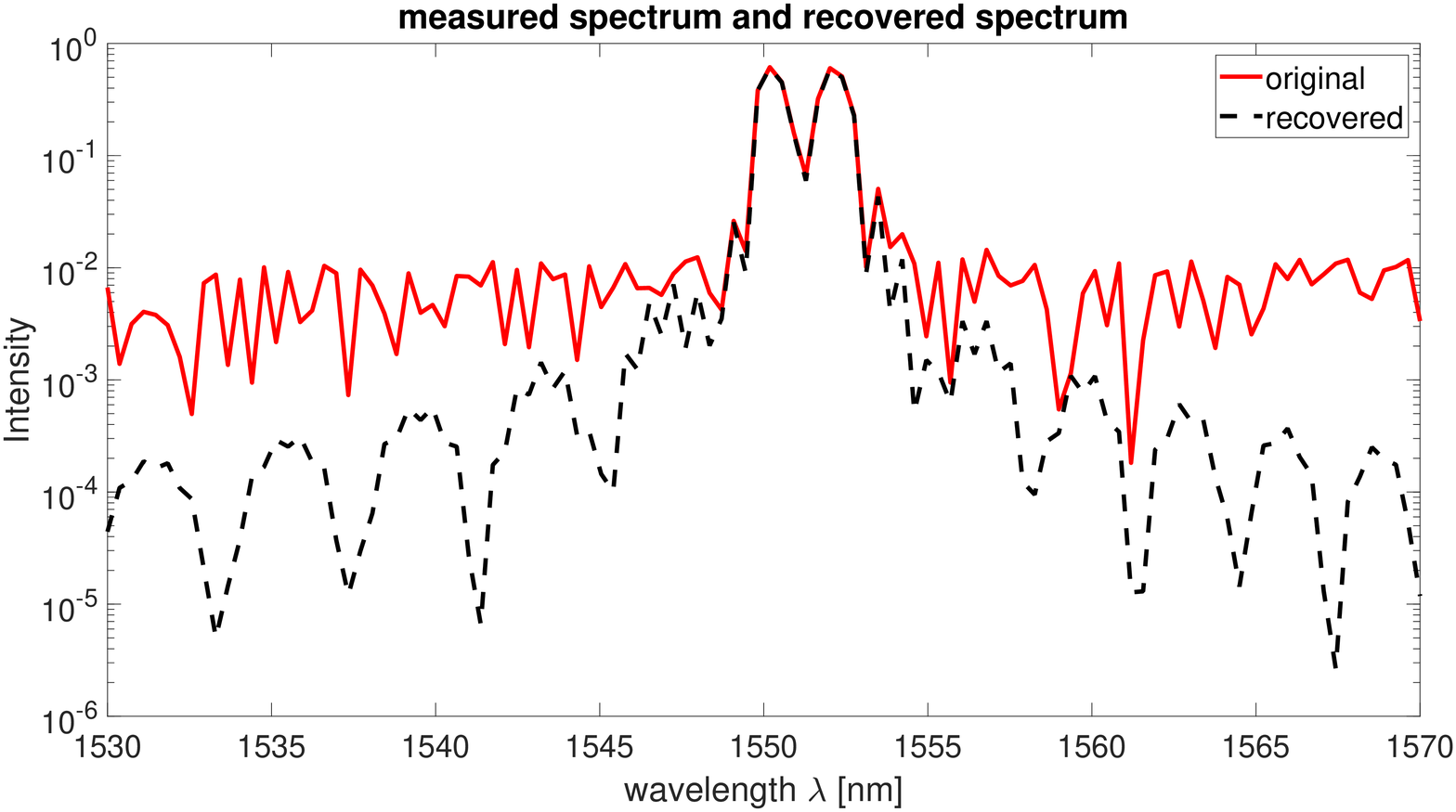}
\includegraphics[width=0.32\linewidth]{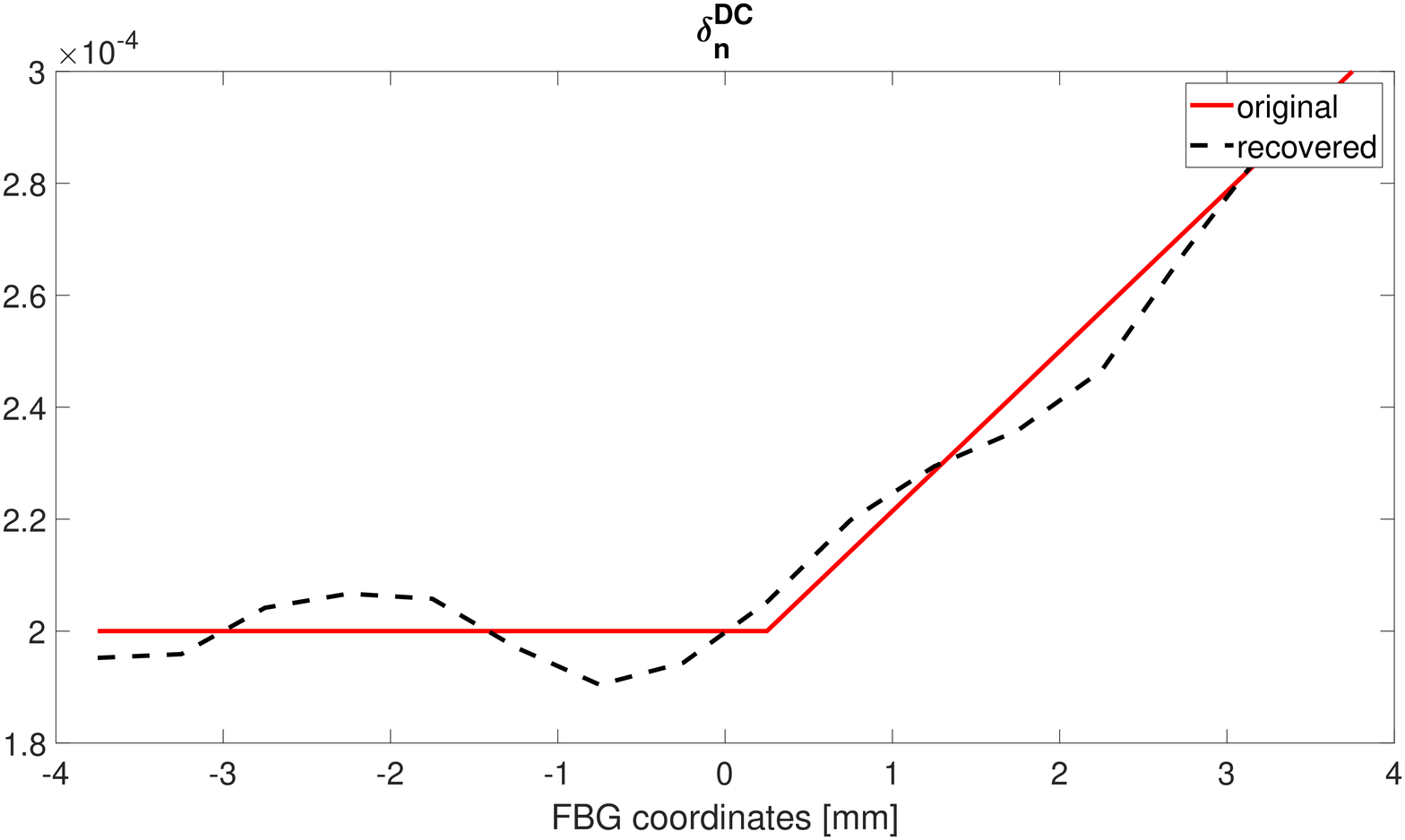}
\includegraphics[width=0.32\linewidth]{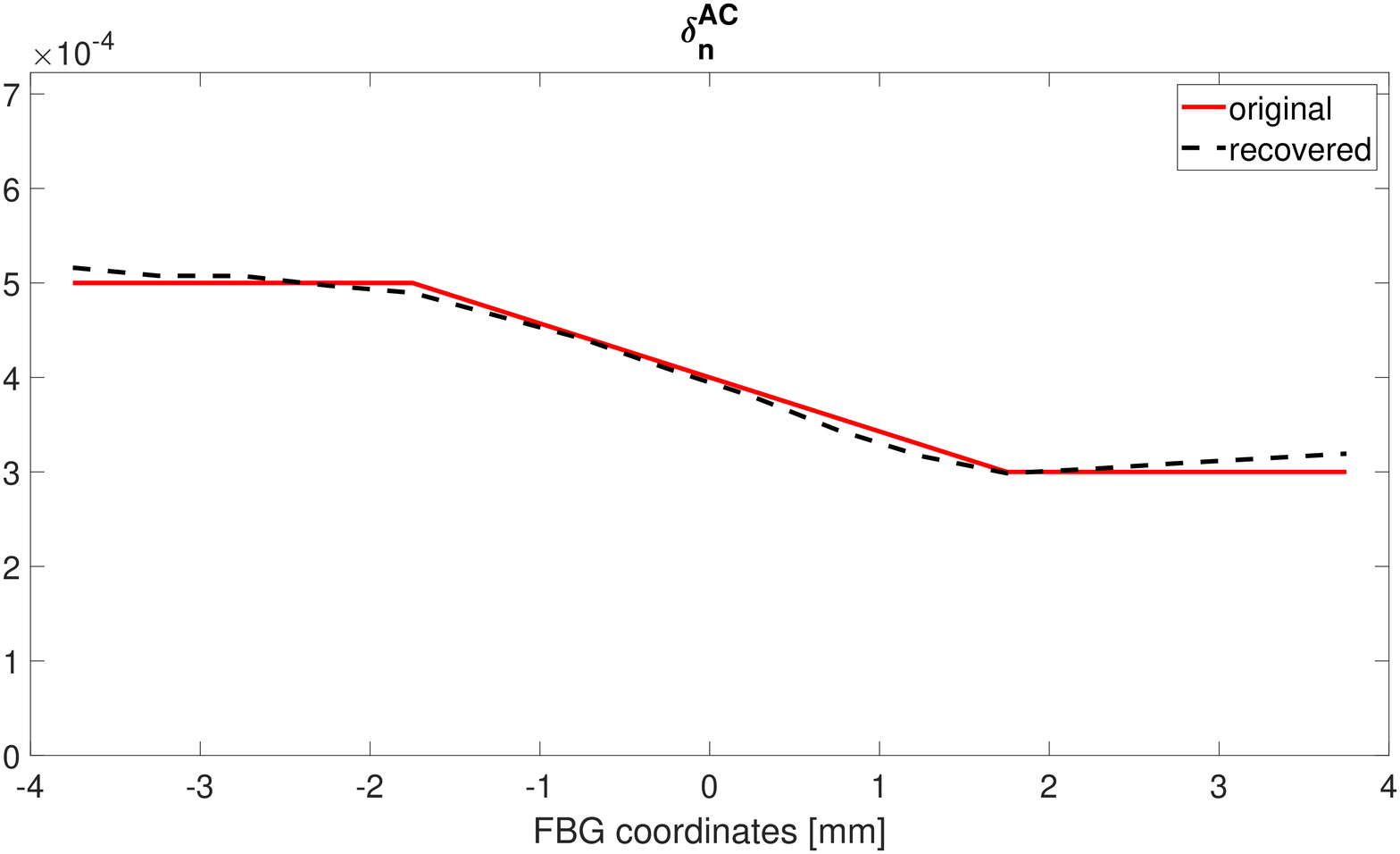}
\caption{Experiment of Section \ref{ssec:apodized}. Recovery of $\dndc$ and $\dnac$ from simulated spectra with $1\%$ relative Gaussian noise. Top row: no regularization, bottom row: regularized solutions with $\alpha=10^{-8}$ and $\beta=10^{-9}$. In both cases the spectra (left column) fit well. The unregularized solutions, however, are highly oscillating and incorrect. Appropriate regularization can overcome this issue and yields a good approximation of the exact solution.}
\label{fig:dndcdnac}
\end{figure}

\subsection{Chirped FBG} \label{ssec:chirp}

So far, we have not considered the FBG period modulation $\Phi$ as an unknown. 
We turn to this case in this section, where for now we let the other parameters $\dndc$ and $\dnac$ be known. To the best of the authors' knowledge, this is the only case for which an inversion of the entire FBG-spectrum has previously been attempted. 
Indeed, in \cite{GillPeters2004} a genetic algorithm was presented that could approximate the local grating period sufficiently well. 
We remark that the authors of \cite{GillPeters2004} used a slightly different encoding of local period changes than done in this paper. 
We also remark that genetic algorithms, and the closely related evolution-type algorithms have to be treated with utmost care when used to solve inverse problems, see the extensive study in \cite{DGMBI} in another context of phase retrieval under a nonlinear forward operator,

In addition the calibration of chirped FBGs (cf. Figure \ref{fig:FBGtypes}), the determination of $\Phi$ can  be directly related to axial strain along the fibre direction. 
Specifically, one has
\begin{equation}\label{eq:phi-strain}
\Phi(z)=\frac{-2\pi z}{\Lambda}\frac{(1-p_\epsilon)\epsilon_z(z)}{(1+(1-p_\epsilon)\epsilon_z(z)}
\end{equation}
where $\epsilon_z$ is the axial strain and $p_\epsilon$ the strain-optic coefficient of the FBG \cite{GillPeters2004,peters2001embedded}. 

We remark that the the period modulation $\Phi$ enters the governing ODE system \eqref{eq:coupledDE_erdogan} only via its derivative $\Phi^\prime$ (cf. the parameter $\hat\sigma$ in \eqref{eq:sigmahat}). Hence, we reconstruct the function $\Phi^\prime$ and obtain $\Phi$ itself via numerical integration using the trapezoidal rule.

As in the previous section a direct recovery approach without regularization fails, as highly oscillating functions produce almost identical spectra as the true reconstruction. 
To suppress this, we employ again a Tikhonov-type approach
\begin{equation}\label{eq:tikh_phi}
[{\Phi^\prime}^\ast(z)]=\argmin_{\Phi^\prime\in(\R^{16}\times \R^{16})} \|F[\Phi^\prime](\lambda)-\rmeas(\lambda)\|_2^2+\gamma \|\Phi(z)^{\prime\prime}\|_2^2
\end{equation}
and obtain $\Phi$ by integration,
\begin{equation}\label{eq:phirec}
\Phi_c(z)=\int_0^z \Phi^\prime(z)\,dz+c.
\end{equation}
If the integration constant $c$ is determined, the strain $\epsilon(z)$ can be calculated from \eqref{eq:phi-strain} as
\begin{equation} \label{eq:epsrec}
	\epsilon(z) = \frac{\Phi_c(z)\Lambda}{(p_\epsilon-1)\Phi_c(z)\Lambda-2\pi z(1-p_\epsilon)}.
\end{equation}
The determination of an integration constant typically relies on some additional information, which in general is not available in our setting. 
Fortunately, we observed a kind of self-calibration effect that can be used to find $c$. 
Indeed, by default we obtain $\Phi_0$ from \eqref{eq:phirec} with $c=0$. 
We then calculate the strain $\epsilon$ according to \eqref{eq:epsrec} for a range of $\Phi_c$, $c\in\R$, and with this the norms $\|\epsilon(z,c)\|_2$ as function of $c$. 
This function has a single minimum whose location coincides with the desired integration constant, see Figure~\ref{fig:phi_c} for examples with noise-free and noisy data..

\begin{figure}
\includegraphics[width=0.49\linewidth]{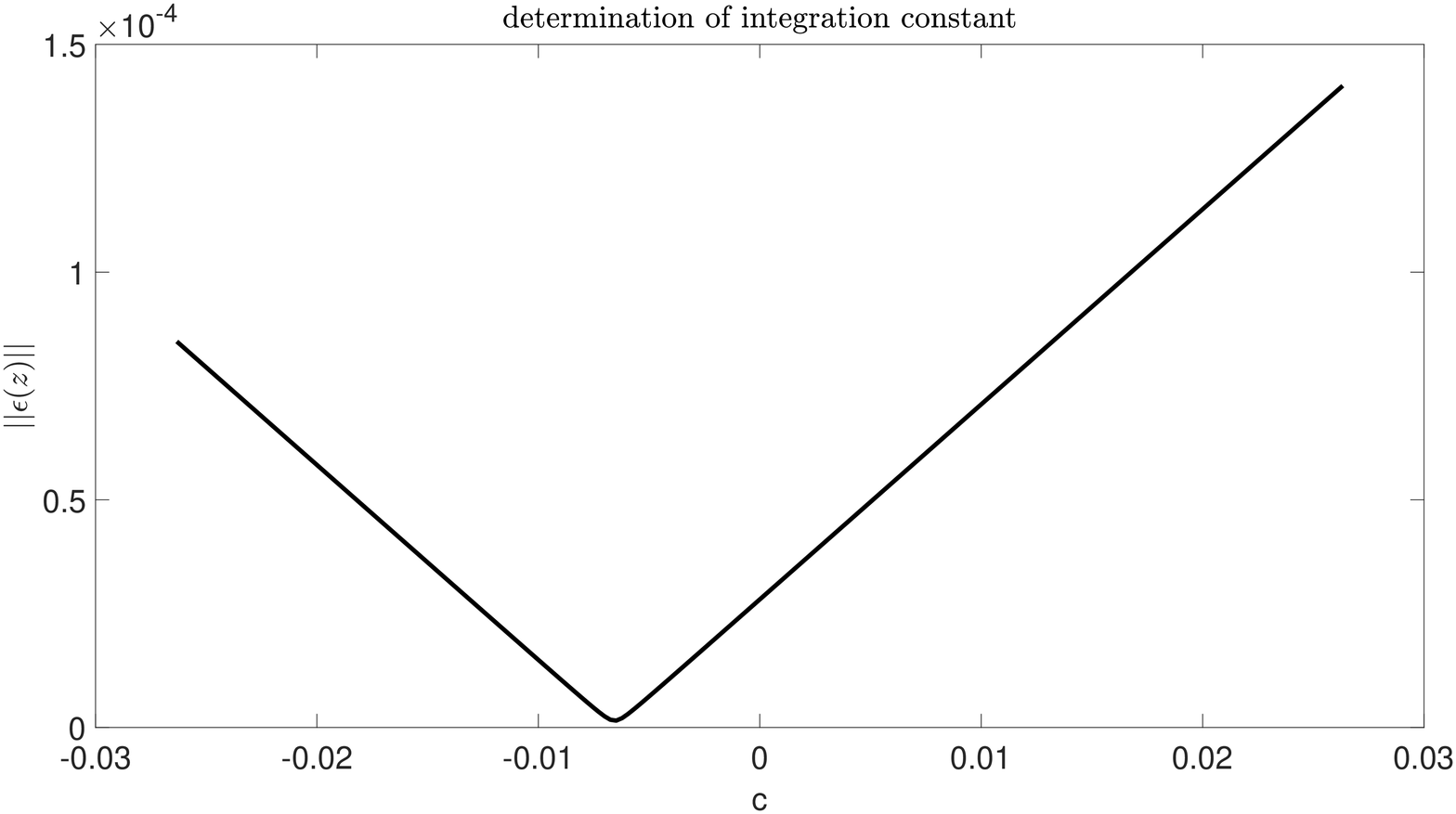}
\includegraphics[width=0.49\linewidth]{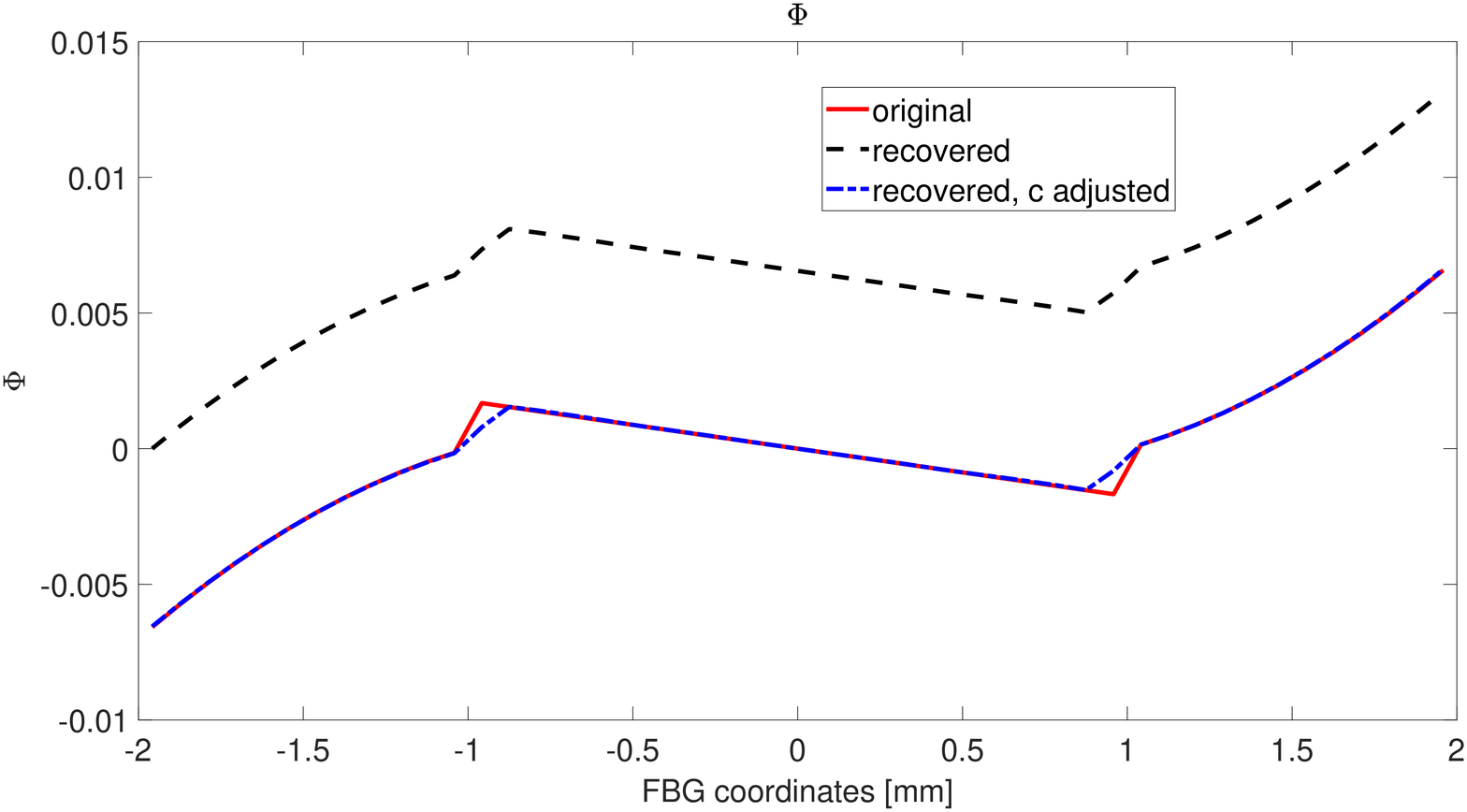}\\
\includegraphics[width=0.49\linewidth]{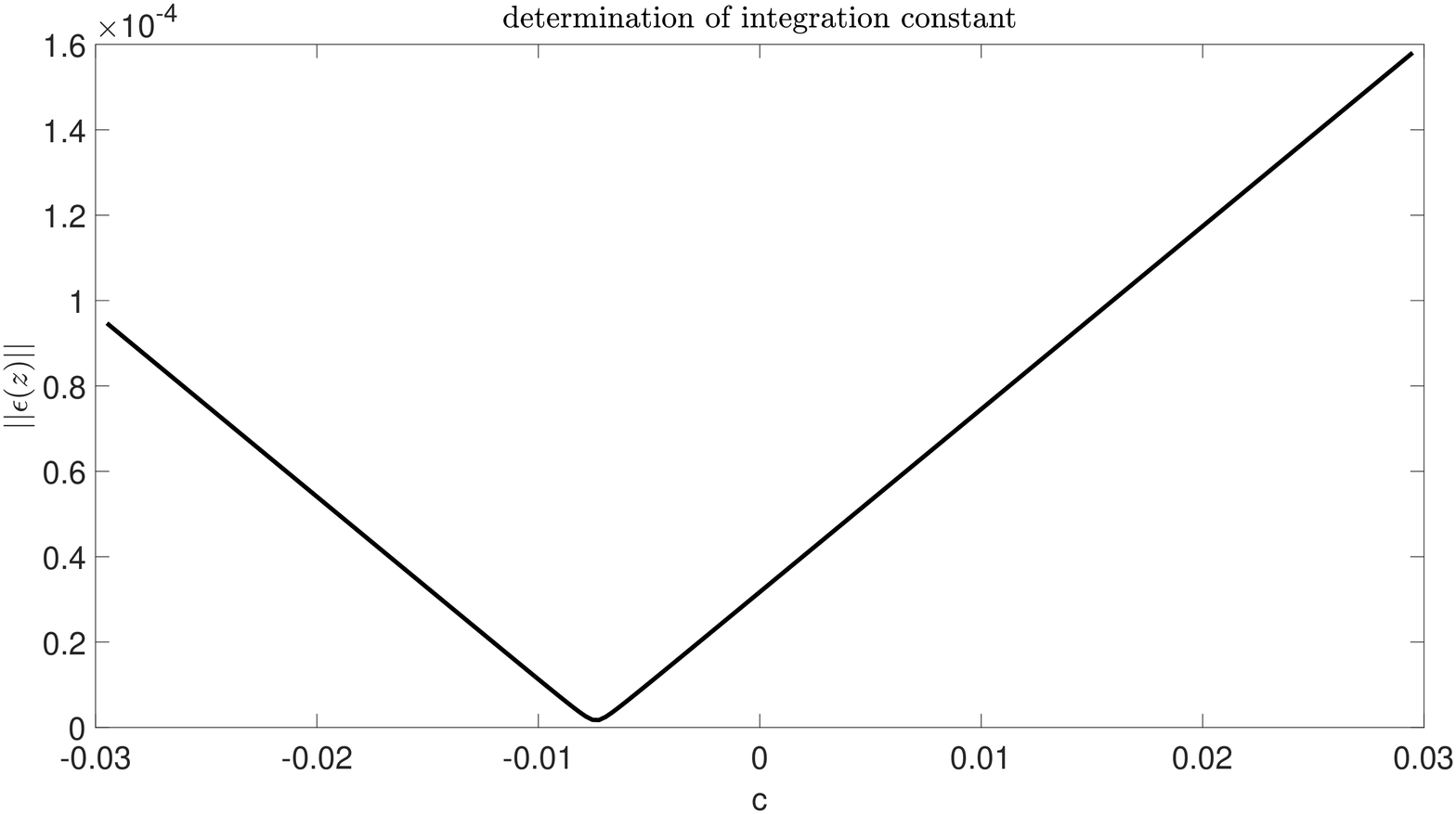}\includegraphics[width=0.49\linewidth]{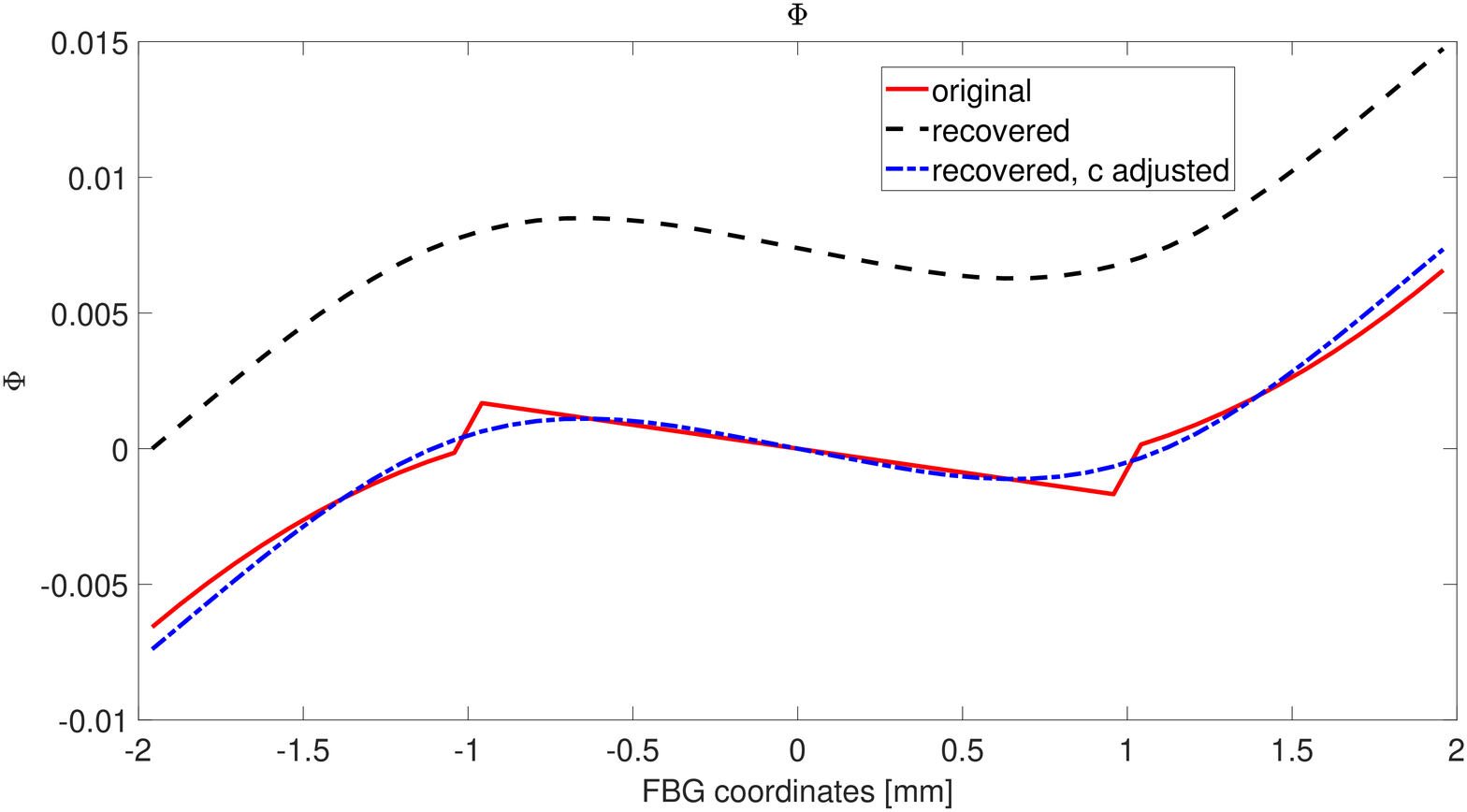}\caption{Estimation of the integration constant in \eqref{eq:phirec} without data noise (top row) and with 1\% data noise (bottom row). Left column: We observe that the norm of $\epsilon$ is a function of the integration constant in \eqref{eq:phirec}. The minimizer corresponds to the correct phase $\Phi$ obtained from \eqref{eq:epsdef} via \eqref{eq:phi-strain}. Right column: A comparison of the true function $\Phi$, the straight-forward solution $\Phi_0$, and the result with $c$ corresponding to the minimum in left plot.}\label{fig:phi_c}
\end{figure}

To present an example for this scenario we attempted to reproduce Example~3 from \cite{GillPeters2004}. 
There the strain function over an FBG of length $L=4\mathrm{mm}$ is defined as
\begin{equation} \label{eq:epsdef}
	\epsilon(z)
	=
	\begin{cases} 
		0.4\cdot10^{-3}z+0.4\cdot10^{-6}  & -0.002\leq z<-0.001\\ 
		0.2\cdot10^{-6}                   & -0.001\leq z\leq 0.001\\ 
		-0.4\cdot10^{-3}z+0.4\cdot10^{-6} & \;\;\; 0.001< z\leq 0.002
	\end{cases}.
\end{equation}
This strain function and its reconstruction with the method described above, as well as the simulated and recovered spectra, are shown in Figure \ref{fig:epsresults_nonoise} for noise-free data and in Figure \ref{fig:epsresults_noise} for noisy data. 
It is important to note that, comparing our spectral data with that in \cite{GillPeters2004}, it appears that our wavelength discretization is significantly coarser (we use an equidistant grid with $0.175\mathrm{nm}$ spacing). 
On the other hand, we split the FBG into $48$ subgratings for the transfer matrix approach compared to $8$ in the original example. 
In particular, we obtain a spatial strain resolution of $83\mathrm{\mu m}$ from only one FBG sensor. 
Similarly to the last section, we solve \eqref{eq:tikh_phi} for several values $\gamma$. 
In the absence of noise in the data, only a small amount of regularization is required. 
We obtain the best reconstruction for $\gamma=1\cdot 10^{-13}$; see Figure~\ref{fig:epsresults_nonoise}. 
We then repeat the experiment with 1\% Gaussian noise added to the data. 
We still obtain  good approximation of the unknown strain $\epsilon$, see Figure~\ref{fig:epsresults_noise} for $\gamma=9.5\cdot10^{-7}$. 
However, the approximations are sensitive with respect to the regularization parameter $\gamma$, and similar to the results in the previous section, it remains an open problem to find an automated way to determine an appropriate regularization parameter. 
Several established heuristic parameter selection rules have been tested, including the quasi-optimality criterion \cite{QO}, the L-curve method \cite{Hansen1993}, the heuristic discrepancy principle \cite{EngHanNeu96}, and the simple L-curve \cite{KindermannRaik}. 
The reason for this failure is likely the nonlinearity of the forward operator, as heuristic parameter selection rules are often developed for linear forward operators; as well as the previously noticed flatness of the residuals w.r.t.\ the regularization parameters. 
In Figure~\ref{fig:phi_res:errors}, we show residuals $\|F[\Phi^\prime]-y^\text{meas}\|$ and reconstruction errors in phase $\|\Phi_\text{rec}-\Phi^\dag\|$ and strain $\|\epsilon_\text{rec}-\epsilon^\dag\|$. 
It can be seen that without simulated data noise, the residuals fall continuously with decreasing $\gamma$ until they reach a flat plateau. 
The minimal reconstruction error is to be found in the region of descending residuals. 
However, under noisy data this is no longer the case, as the residuals are almost unchanged for the range of regularization parameters shown. 
Also the reconstruction errors in $\Phi$ and $\epsilon$ are flat, and it is difficult to locate the minima. Again, this indicates the need for adaptive parameter selection strategies, which remains a topic for future research. 

\begin{figure}
\includegraphics[width=0.49\linewidth]{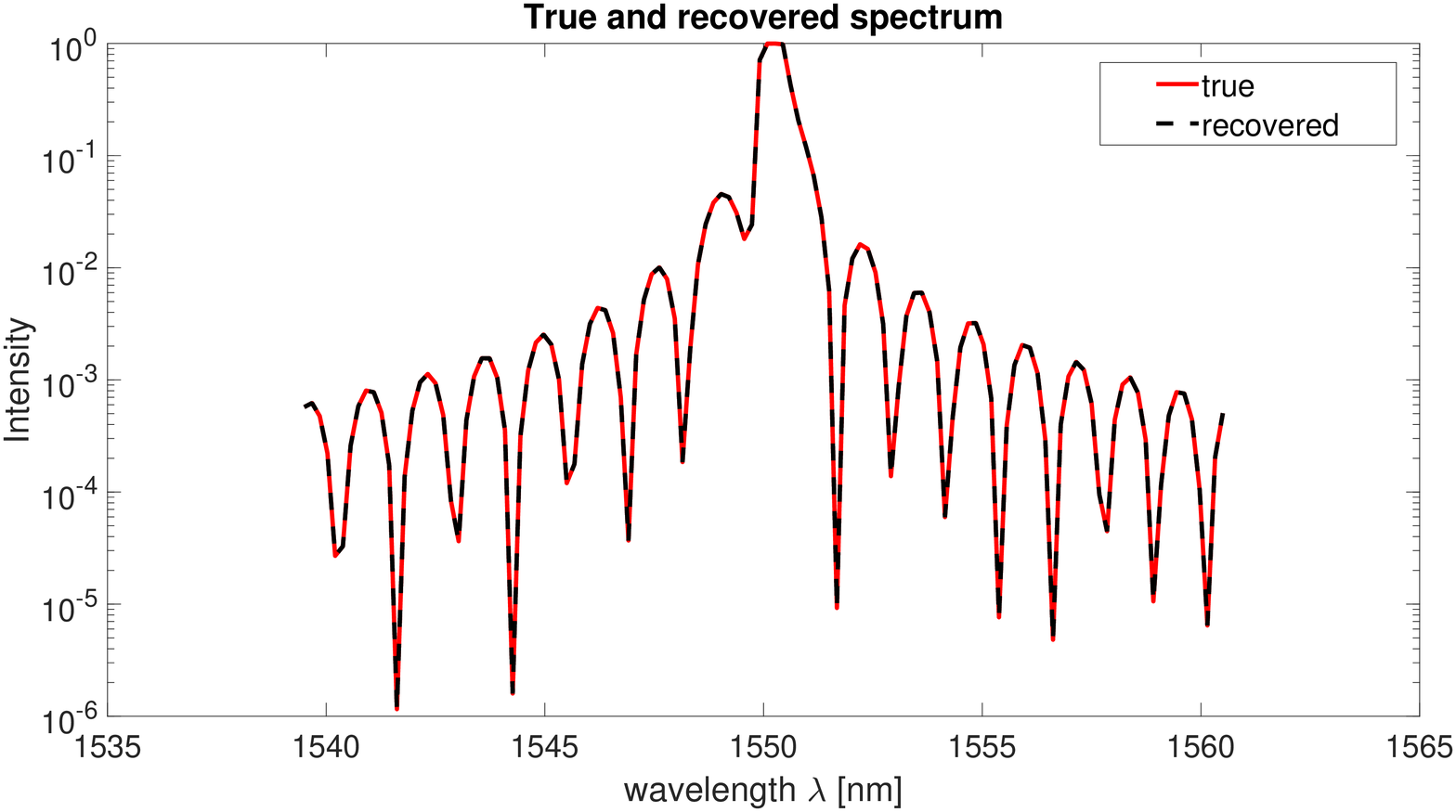}\includegraphics[width=0.49\linewidth]{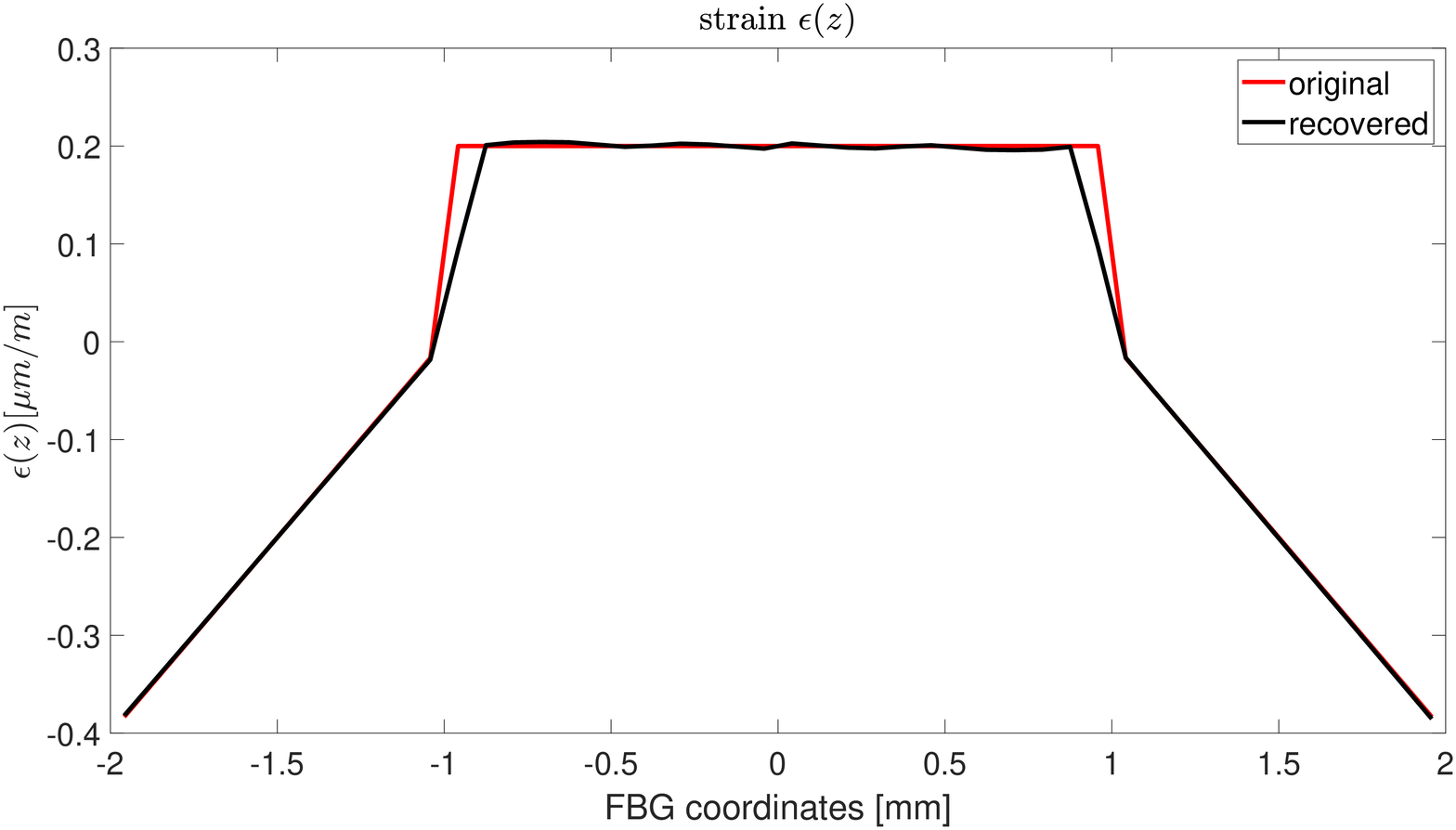}
\caption{Adaption of \cite[Example 4.3]{GillPeters2004}. Noise-free data. 
Left: Simulated and reconstructed spectrum. 
Right: true and recivered strain.}
\label{fig:epsresults_nonoise}
\end{figure}

\begin{figure}
\includegraphics[width=0.49\linewidth]{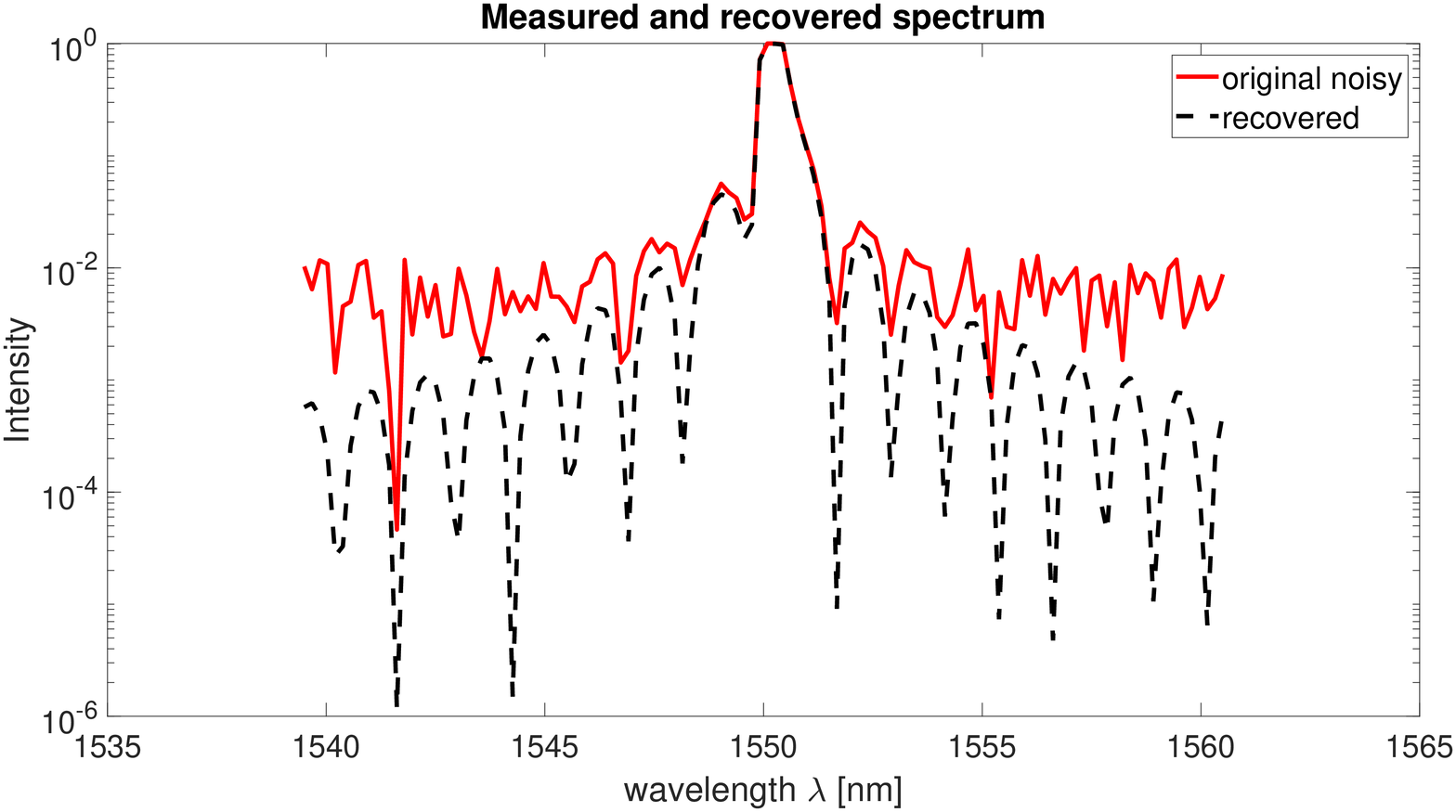}
\includegraphics[width=0.49\linewidth]{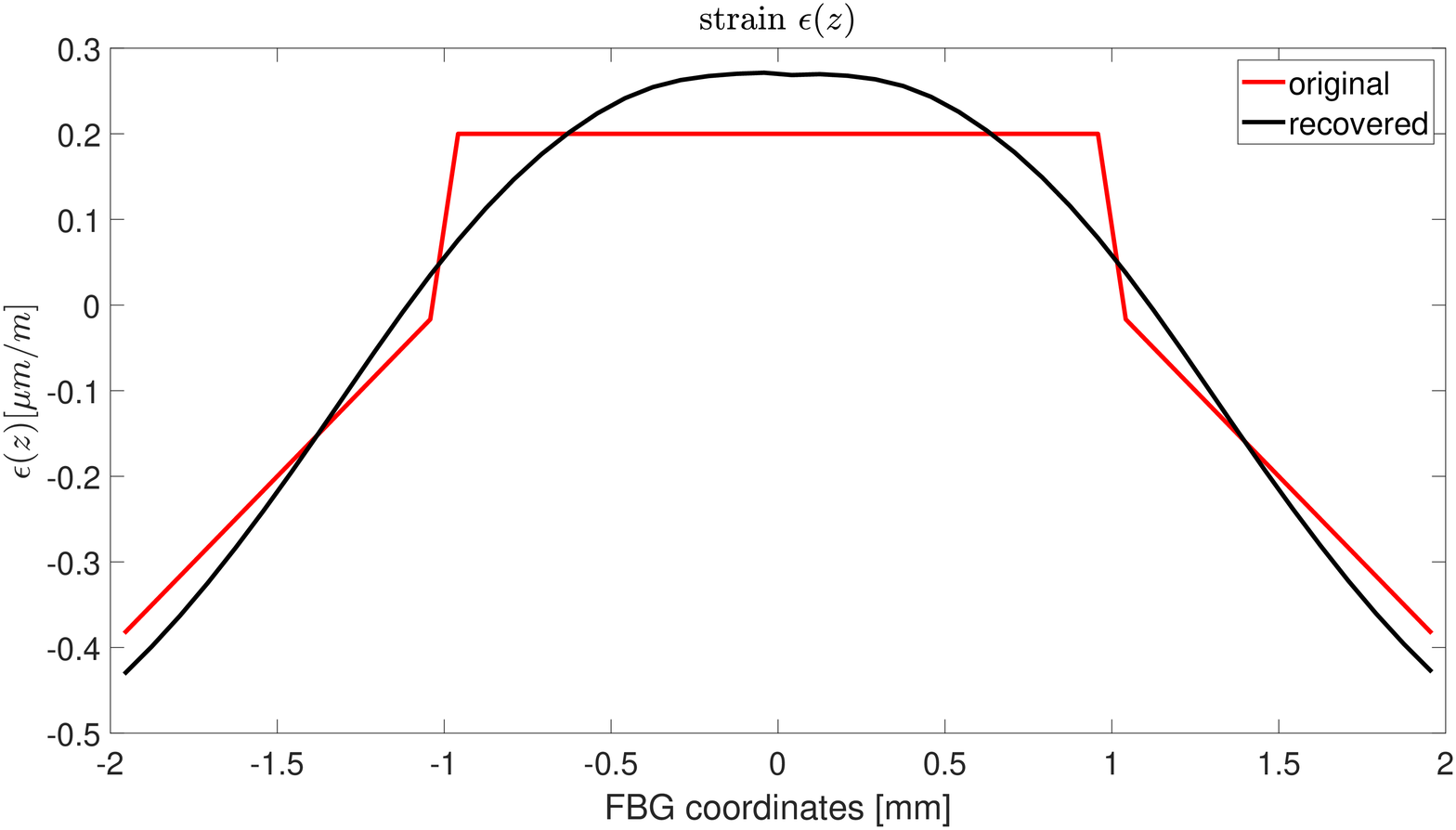}
\caption{Adaption of \cite[Example 4.3]{GillPeters2004}. $1\%$ noise added to the simulated data. 
Left: Simulated and reconstructed spectrum. 
Right: true and recovered strain.}
\label{fig:epsresults_noise}
\end{figure}

\begin{figure}
\begin{tabular}{c c}
\includegraphics[width=0.49\linewidth]{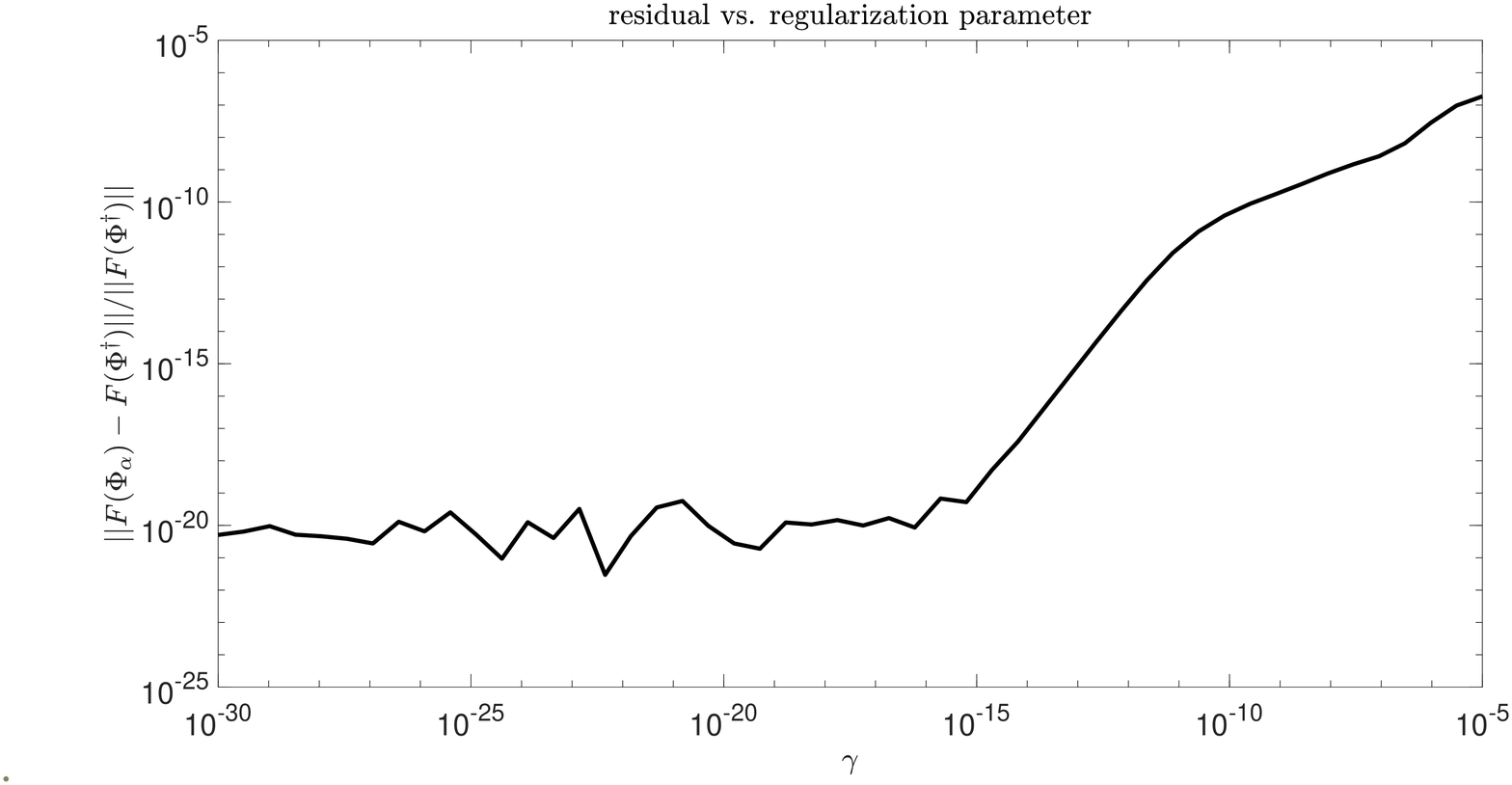} & 
\includegraphics[width=0.49\linewidth]{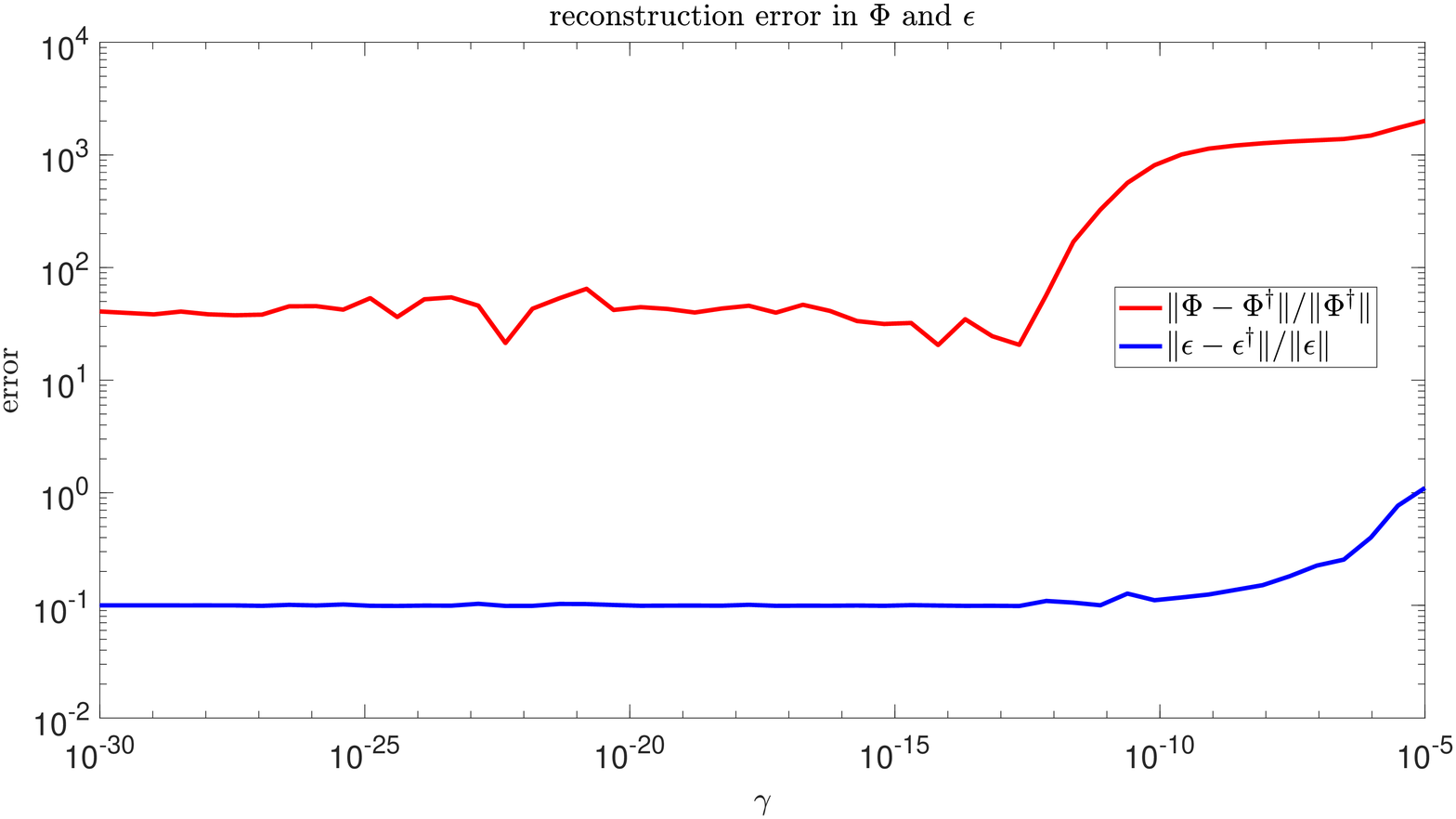} \\
\includegraphics[width=0.49\linewidth]{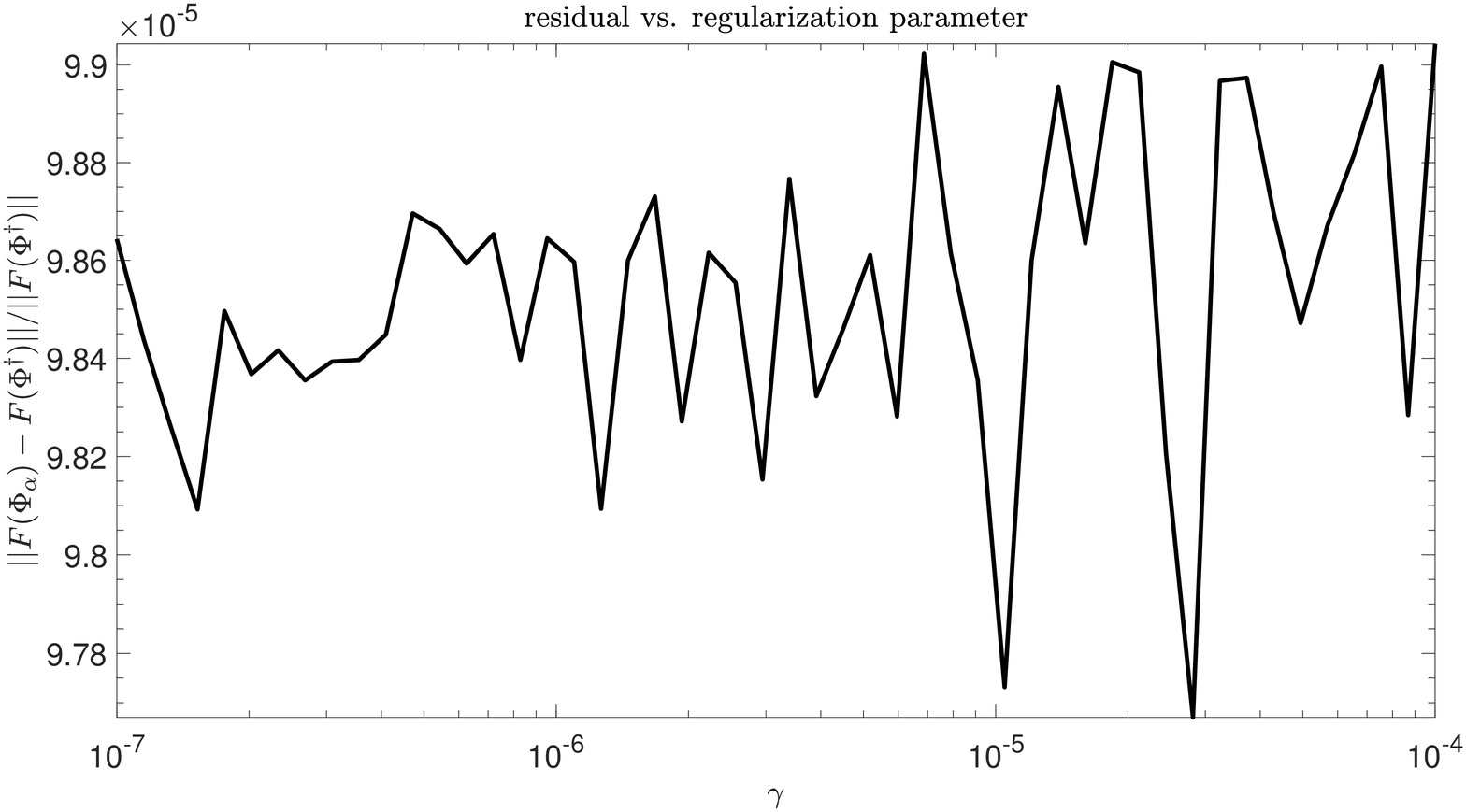} & 
\includegraphics[width=0.49\linewidth]{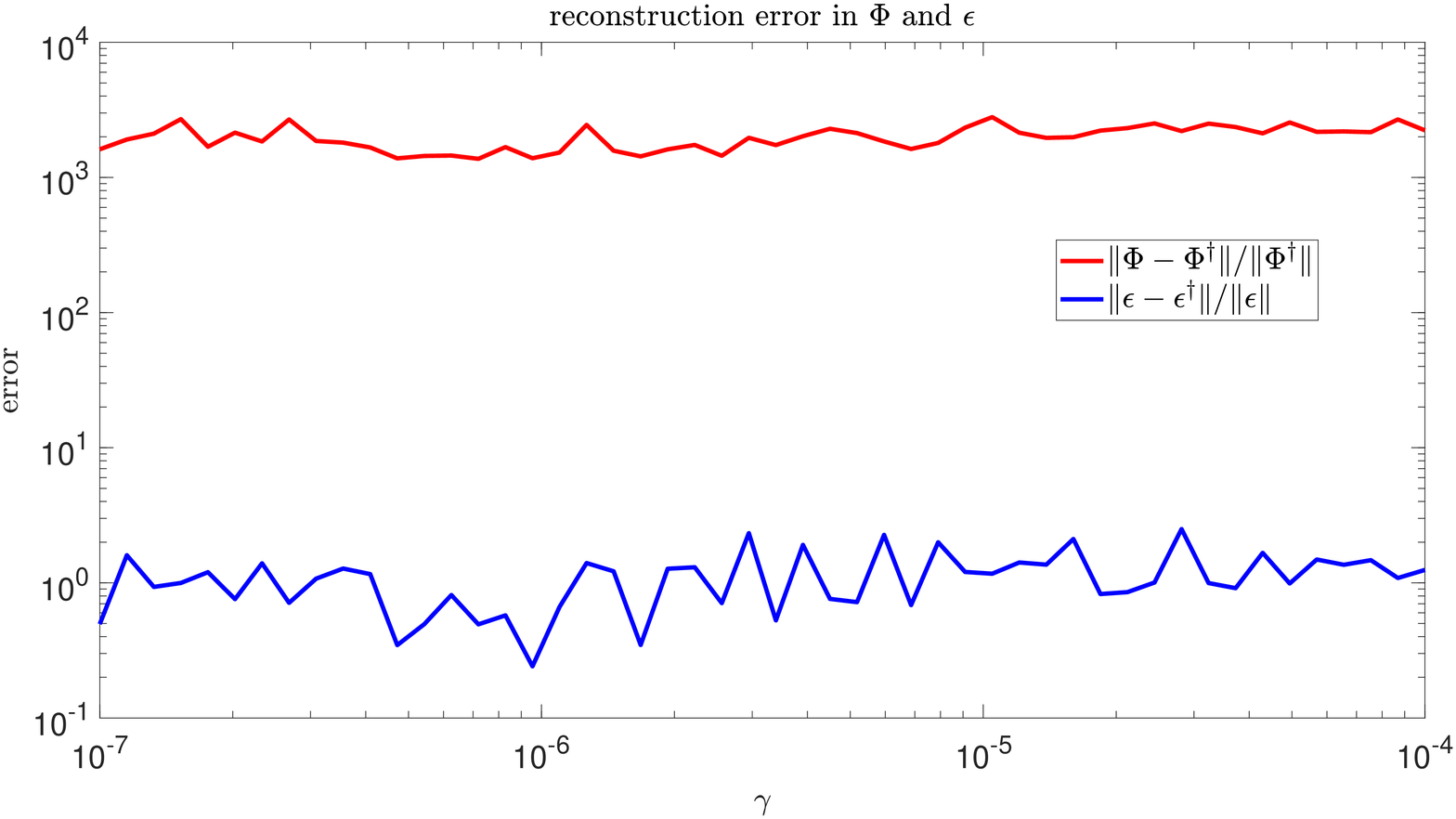}
\end{tabular}
\caption{Residuals (left column) and reconstruction errors (right column) for noise-free data (top row) and with 1\% noise (bottom row) plotted against the regularization parameters $\gamma$. 
The vertical line indicates the least reconstruction error in the strain. 
Except for large values of $\gamma$, the curves are on one hand highly oscillating, but on the other hand also quite flat. 
Hence, well established heuristic parameter selection rules appear not to be applicable for finding an appropriate value of $\gamma$.}
\label{fig:phi_res:errors}
\end{figure}

In summary, we have demonstrated that it is possible to recover a strain distribution over the length of an FBG by minimizing a Tikhonov-type functional with a spatial resolution below 100$\mu$m. The reconstruction quality depends strongly on the choice of the reconstruction parameter, and it remains an open problem to find a suitable rule for its selection.

\section{Experiments} \label{sec:experiments}

In this section we demonstrate the regularization approach with real measured data. 
We begin with the simplest test setting, a homogeneous deformation of the FBG. 
When using FBG sensors, a narrowband peak is a basic requirement for subsequent strain determination. 
A narrowband peak allows a main reflected wavelength to be determined. 
Various methods have been proposed in the literature for this purpose \cite{Hoffmann2008Grundlagen}, e.g., by the centroid algorithm, the centre wavelength, a spline reconstruction or the full-width-at-half-maximum method (FWHM).

This reflected peak wavelength can now be determined at different points in time. 
Based on a change of the reflected wavelength, different physical quantities can be directly inferred. 
In our case this is the mechanical strain, which is obtained by solving for $\varepsilon_z$ in \eqref{eq:strain}, giving
\begin{equation} \label{eq:strain2}
	\varepsilon_z = \frac{1}{K} \frac{\Delta \lambda_B}{\lambda_B} .
\end{equation}

This principle works well if the reflected peak is narrowbanded. 
In the literature two main effects are described where the reflected peak loses its narrowbanded form. 
These two effects, for which the peak widens or splits, are transversal and inhomogeneous loading of the sensor. 
Transversal loading (birefringence) is described in \cite{Guemes2002} and \cite{Lawrence1999}, see also Remark \ref{rem:biref}. 
This case is important for complex loading cases in complex components or reinforced materials, especially unidirectional reinforced materials. 
Currently this is not included in our model.

The second main effect is the inhomogeneous loading of a FBG sensor. 
This inhomogeneous case can occur in different situations. 
Foremost, the loading of a FBG sensor depends on the load case. 
Another source are geometric inhomogeneities, which also lead to inhomogeneous loading. 

\subsection{Experimental Setup and Specimen Design} \label{ssec:expSetupDesign}

In this paper we have intentionally created a geometric inhomogeneity to study the behaviour of FBG sensors and to test the inversion approach. 
Specifically, an aluminium bending beam was selected with linear-elastic material behaviour. 
The aluminium beam is loaded in a 4-point bending test, shown in Figure~\ref{fig:Versuchsaufbau}, as this produces an analytically well identifiable stressing in the region between the inner bearings, where the shear force is zero and the bending moment is constant.

In order to investigate the required magnitude of the inhomogeneity for producing an effect on the reflected peak, different specimens were produced, which are displayed in Figure~\ref{fig:Kerbgeometrie}. 
Different notches were manufactured in the aluminium beams. 
The geometry of these notches was determined by finite element simulation in Ansys. 
The notches were designed in such a way as to produce a specific characteristic of the strain along the FBG measuring length. 
This resulted in three variants of the test specimen. 
The first specimen shows a small strain gradient. 
The second specimen generates a medium strain gradient at the eccentric sensor of 30\,\% strain difference over the sensor length. A large strain gradient is produced by the third specimen, which has a strain difference at the eccentric sensor of 70\,\% over the sensor length. 

\begin{figure}\centering
	\includegraphics[width=1.0\linewidth]{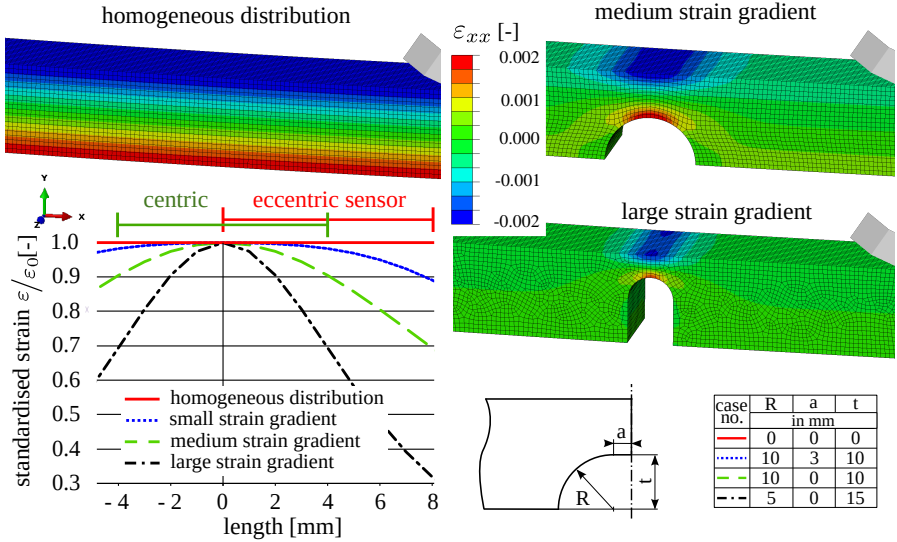}\caption{Different notch designs are identified to generate different strain gradients.}\label{fig:Kerbgeometrie}\
\end{figure} 

\begin{figure}\centering
	\includegraphics[width=1.0\linewidth]{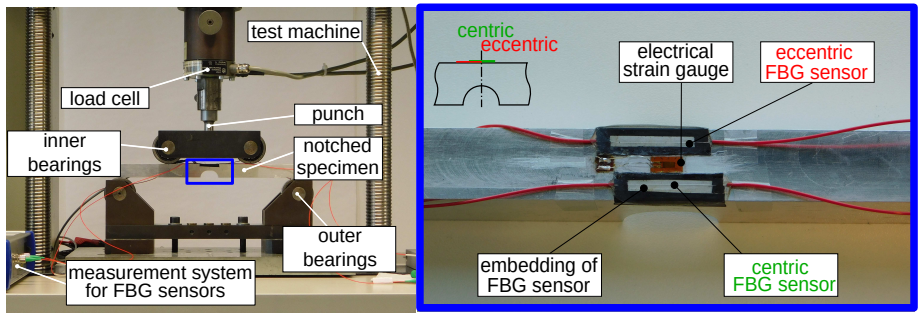}
	\caption{Experimental Setup.}\label{fig:Versuchsaufbau}\
\end{figure}

The left panel of Figure~\ref{fig:Versuchsaufbau} shows the experimental setup of the four-point bending test.  In this experiment, a 100 kN ZwickRoell\textsuperscript{\textregistered} testing machine is used.
Besides a force controlled introduction of a preload, the loading is applied displacement controlled. 
Thereby, the force is measured by a 5\,kN load cell.
The displacement produces an approximated averaged strain of 2000\,$\mu m/m$ on the surface of the specimen. 
For this purpose, the punch is pressed onto a carriage, which introduces the force into the specimen via two mounted rolls (inner bearings). The outer rolls represent the outer bearings and are also rotatable. This means that no lateral force is generated.  The produced, mean strain is measured by the electrical strain gauge with the MGCplus CP22 measurement system from HBM with catman AP V3.5.1.48.  

Three different sensors are applied on each specimen as depicted on the right image of Figure \ref{fig:Versuchsaufbau}. An electrical strain gauge is installed centrally above the notch and centrally at beam depth and serves as a reference sensor. The strain gauge we use, is a sensor 1-LY13-3/250 from HBM. Furthermore, two fibre Bragg grating sensors are applied. TOne edge of the so-called eccentric sensor aligns to the centre of the notch. The second, centric, sensor is centred above the notch. In this case draw tower gratings (DTG\textsuperscript{\textregistered}) from FBGS with a sensor length of 8\,mm were used. The FBG signal is measured by the FBGscan 800 from FBGS with the software program IllumiSense V.2.3.5.

Echoing the remark at the end of Section~\ref{sec:unique}, our currently available measurement setup is not well suited to determine the complex loading configurations in the experiments we carried out. 
We know from the manufacturer that the employed FBGs are homogeneous ( $\dnac=\mathrm{const}$, $\dndc\equiv 0$, $\Phi\equiv 0$) with a length of 8mm. 
Other parameters such as $\Lambda,n_\text{eff}$, or the precise value of $\dnac$ are not provided. 
Furthermore, the obtained spectra have a poor wavelength resolution of approximately 168nm, which alone is not enough for a precise characterization of the peak for the FBG without loading, since the FWHM of the peak is about 450nm. 
(In other words, the measured maximum spectrum  may deviate from the correct peak due to the coarse discretization.) 
Finally, the spectra are not scaled according to the theory (with values up to around 80 for the intensity, which should reach 1 at maximum) and exhibit some background noise at an intensity of approximately 10 for which we have no explanation and which is not consistent with the forward model. 
An example of a measured spectrum  $r^\text{meas}(\lambda)$ is shown in Figure~\ref{fig:crapdata}. 
We also show the rescaled spectrum which we use as input for the reconstruction algorithm: We take $r^\text{meas}$, divide by 100, subtract its minimum, and divide by the estimated peak value. 
To estimate the peak, we make a fit of the modified data $\tilde r^{meas}:=r^\text{meas}(\lambda)/100-\min(r^\text{meas}(\lambda)/100)$ to a Gaussian curve $c_1\exp(-\frac{(\lambda-c_2)^2}{c_3^2})$ to find $c_1,c_2$, and $c_3$, and set $\tilde r^\text{meas}\mapsto \tilde r^\text{meas}/c_1$. 
The rescaled data and an example of a homogeneous fit using the approach of Section~\ref{ssec:hom} is also shown in Figure~\ref{fig:crapdata}, which yields values for $\dnac$ and $\dndc$. 
In order to obtain this, the remaining unknowns are mostly guessed. 
We find the peak wavelength $\lambda_B$ of the unperturbed FBG and set $n_0=1.46$ which is a typical value. This yields $\Lambda_0=\frac{\lambda_b}{2n_0}$. 
Despite the many inaccuracies, we show in the following that even from this ``worst-case'' data significant information about the FBG loading can be extracted. 

\begin{remark} \label{rem:gauss}
One can use the mentioned Gaussian fit directly to assess the shift of the peak in a homogeneous setting, since $c_2$ yields the peak wavelength. 
We demonstrate later that this approach appears to be slightly better than the evaluation of the software built into the interrogator.
\end{remark}
 
\begin{figure}
\includegraphics[width=0.49\linewidth]{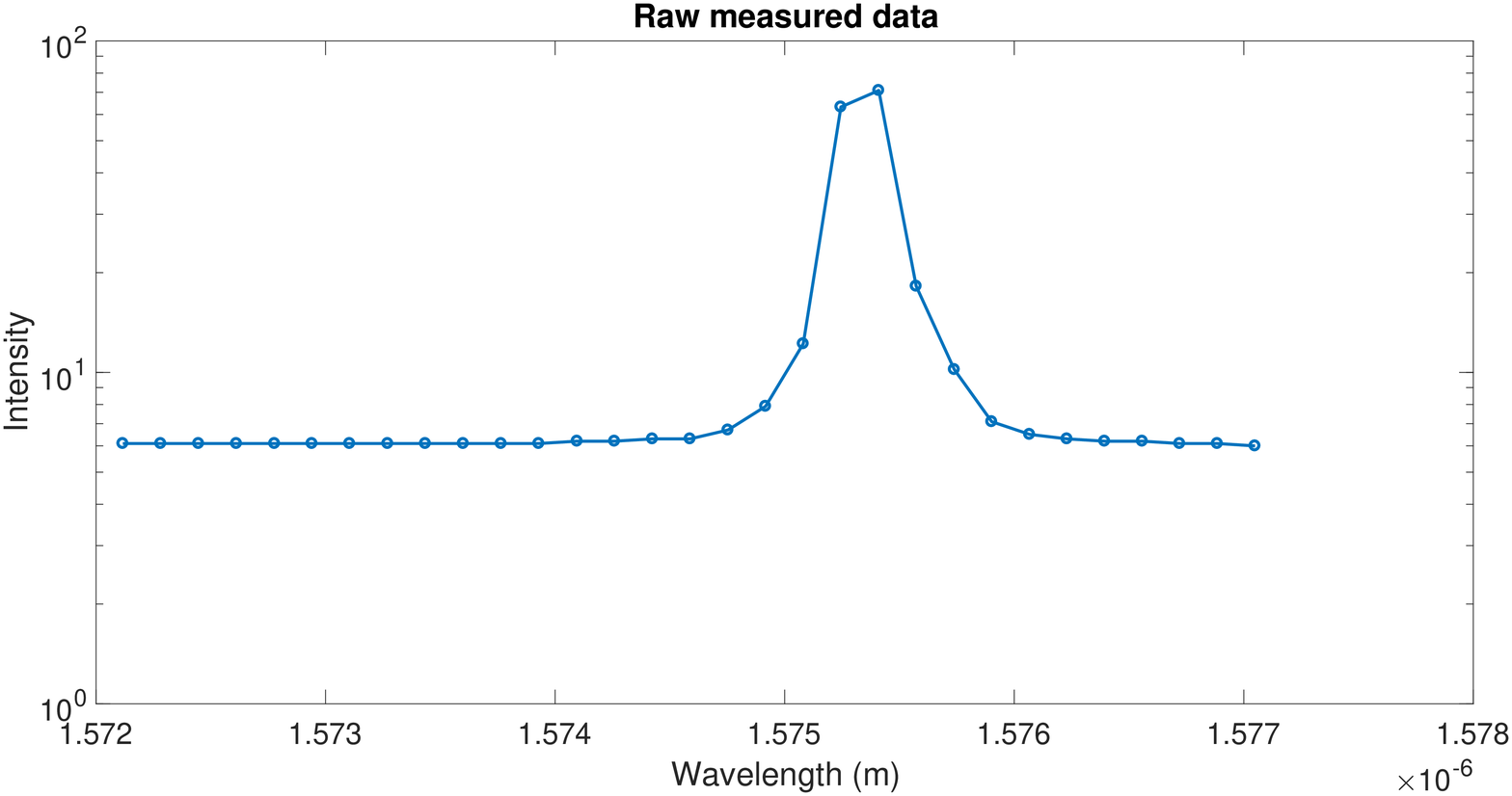}
\includegraphics[width=0.49\linewidth]{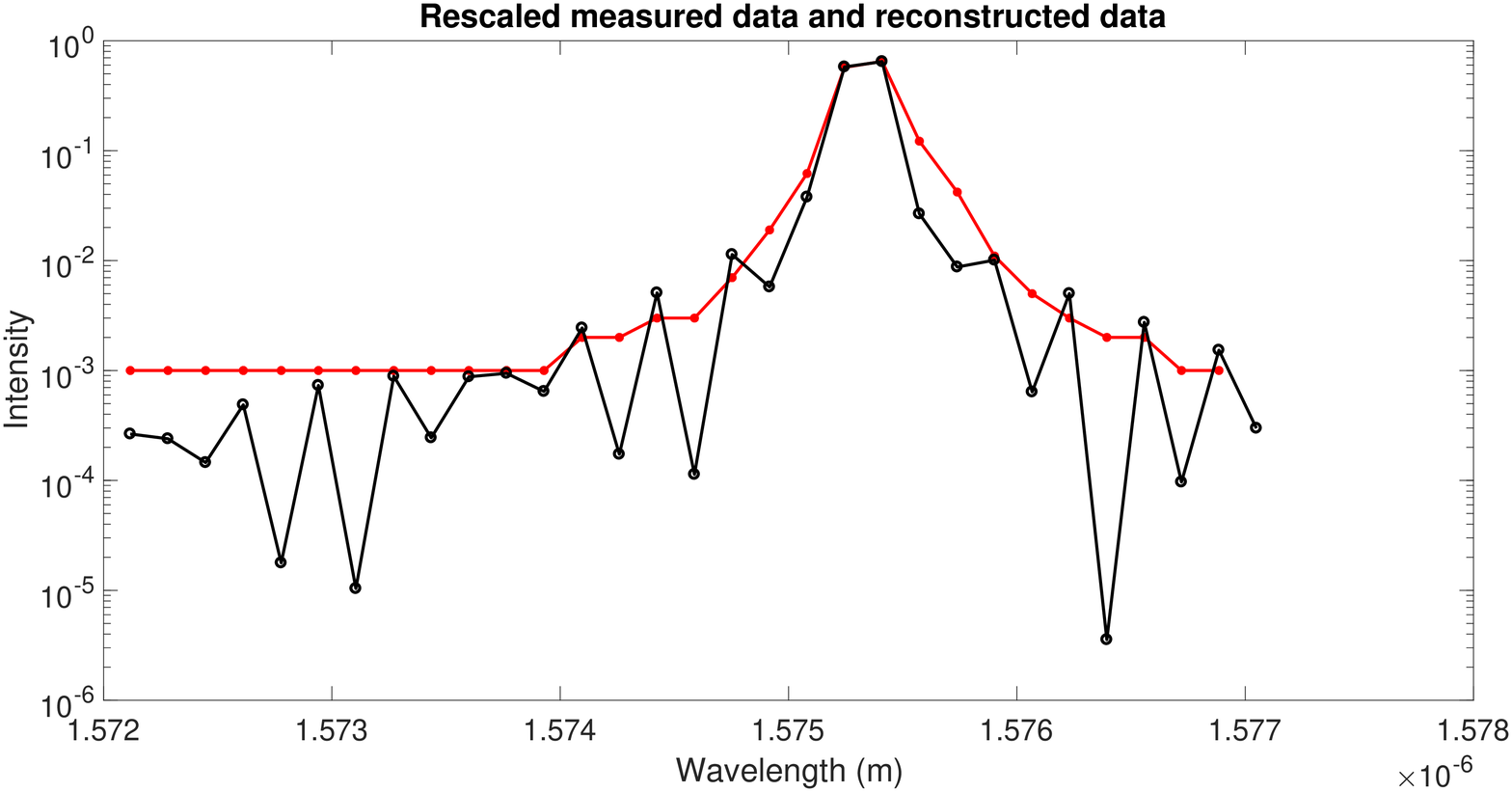}
\caption{Left: raw data marked with \textit{o}. Right: Rescaled data (red) and our fit (black). Our data carries almost no information outside of the peak, which makes a fit challenging.}
\label{fig:crapdata}
\end{figure}

\subsection{Analysis of Experimental Data: Homogeneous Beam} \label{ssec:homogen}

In our first experiment we test the homogeneous setting, with the regularization as discussed in Section \ref{ssec:hom}. 
While this is not the intended application of our method, it serves to demonstrate that the concept works, in particular with respect to the provided data. 
The results for the centric sensor are shown in Figure~\ref{fig:strain_hom}, they are similar for the eccentric sensor. 
We compare four evaluations of the strain: the strain gauge reading, the strain obtained directly from the interrogator, the strain calculated by a Gaussian peak fit (cf.\ Remark~\ref{rem:gauss}), and the strain calculated from the peak shift through the reconstruction of $\dndc$. 
All curves are close together and only differ in details. 
In the zoomed-in image, we see that the original data is not strictly linear as it appears superimposed with a sinusoidal wave of small amplitude. 
The Gaussian peak estimate is able to remove most of this sinusoidal wave and hence may be more accurate. 
The reconstructed curve is, arguably on average, closer to the strain gauge reading, but displays a distinct periodic staircasing effect. 
After several experiments, we attribute this purely to the insufficient data quality. 
Because the wavelength resolution is too coarse, there is an inherent sampling aliasing in the measurement process. 
We were able to recreate this in simulations with coarsely sampled spectra. 
Nevertheless, the resulting strain yields the lowest residuals to the measured data. 
Hence, we conclude that the algorithm itself works as intended.

\begin{figure}\centering
\includegraphics[width=0.49\linewidth]{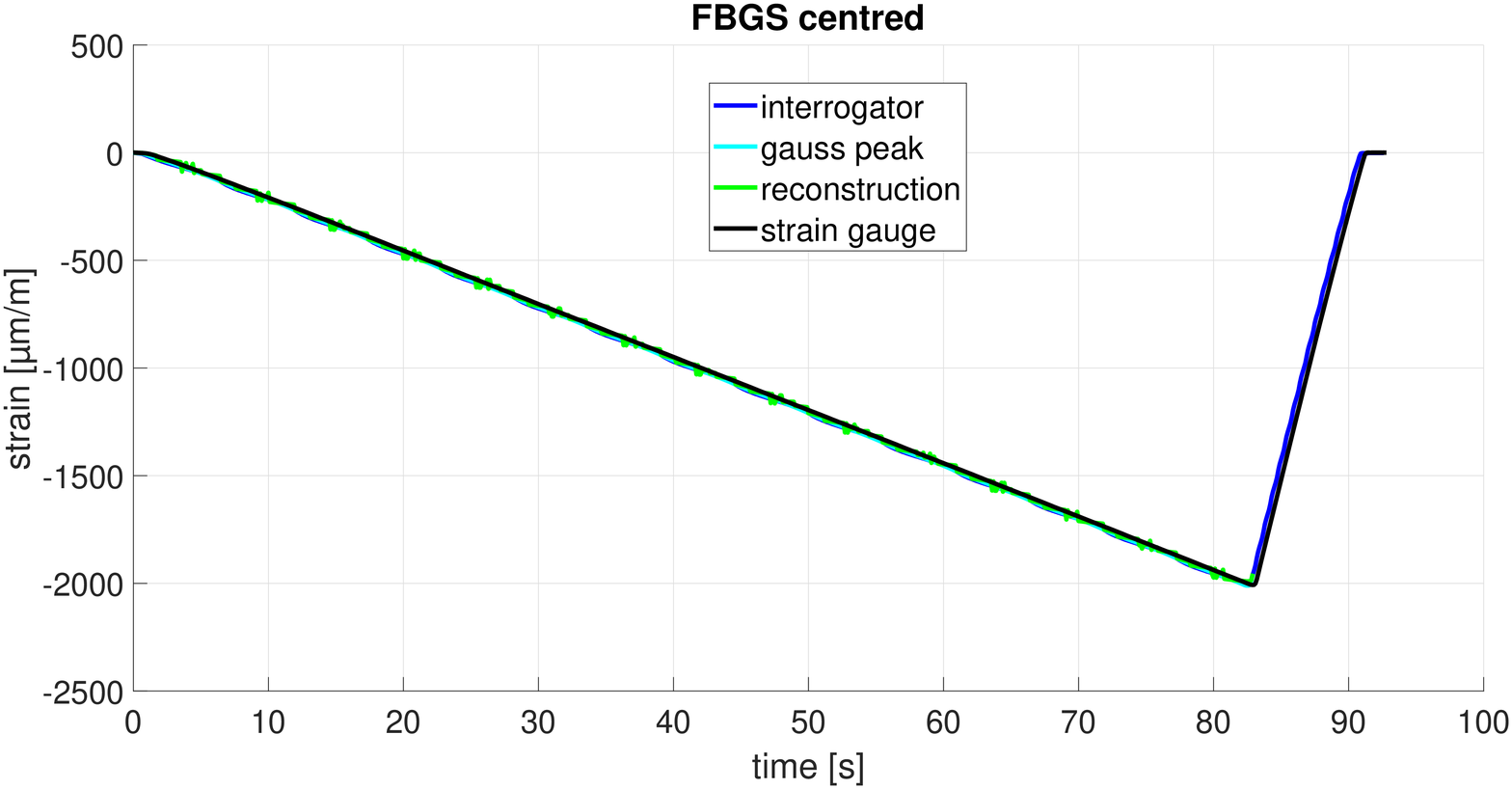}
\includegraphics[width=0.49\linewidth]{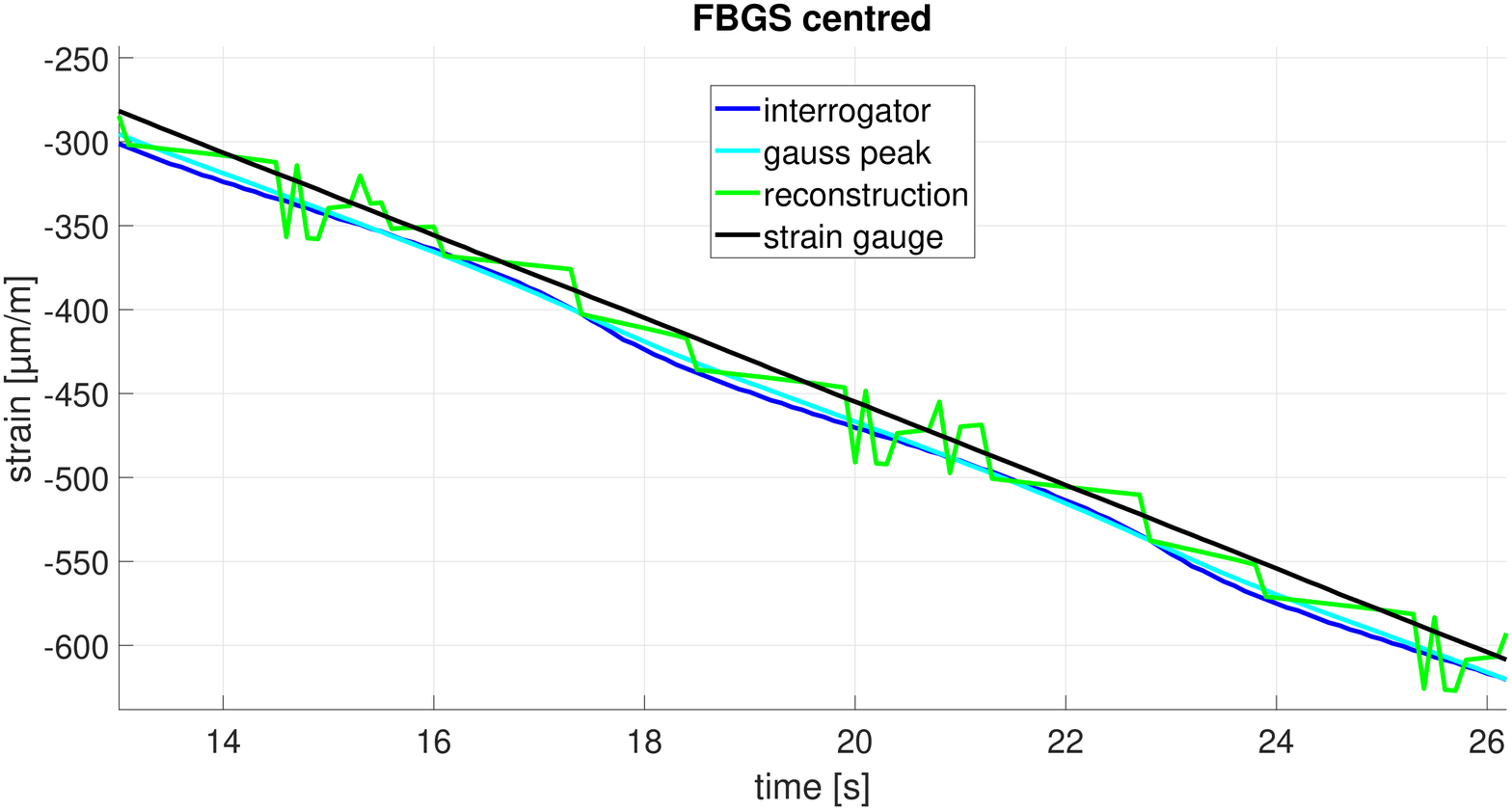}
\caption{Centred FBG under homogeneous strain. Strain gauge (black), strain measured by the interrogator (blue), strain calculated from peak shift through Gaussian fit (see Remark \ref{rem:gauss}, light blue), and strain calculated from the homogeneous reconstruction approach \ref{ssec:hom} (green). Overall all curves fit together. In the zoomed image (right), we see a staircasing effect for the reconstruction approach due to low data quality.}
\label{fig:strain_hom}
\end{figure}

\subsection{Analysis of Experimental Data: Medium Strain Gradient} \label{ssec:middle}

In this experiment we simulate a slight inhomogeneity in the strain applied to the FBG caused by an inserted notch. 
The notch produces a stress concentration which leads to an inhomogeneous deformation state. 
As seen in Figure~\ref{fig:spektren_4er}, the eccentric sensor still sees an (almost) homogeneous loading, and the spectra only show a slight loss in peak height. 
The centric sensor, on the other hand, displays spectra that are considerably wider and of lower peak intensity for larger applied strain values. 
This means the homogeneous model becomes inaccurate. 
To illustrate this, we show in Figure~\ref{fig:strain_mittel} the strain reconstruction obtained from the peak shift in the same way as in the homogeneous case of the previous section. 
Now neither of the FBG evaluation methods follows the strain gauge, simply because the model is incorrect. 
We also added an example of a spectrum and its reconstruction with the homogeneous FBG model, which clearly no longer fit.

\begin{figure} \centering
	\includegraphics[width=1\linewidth]{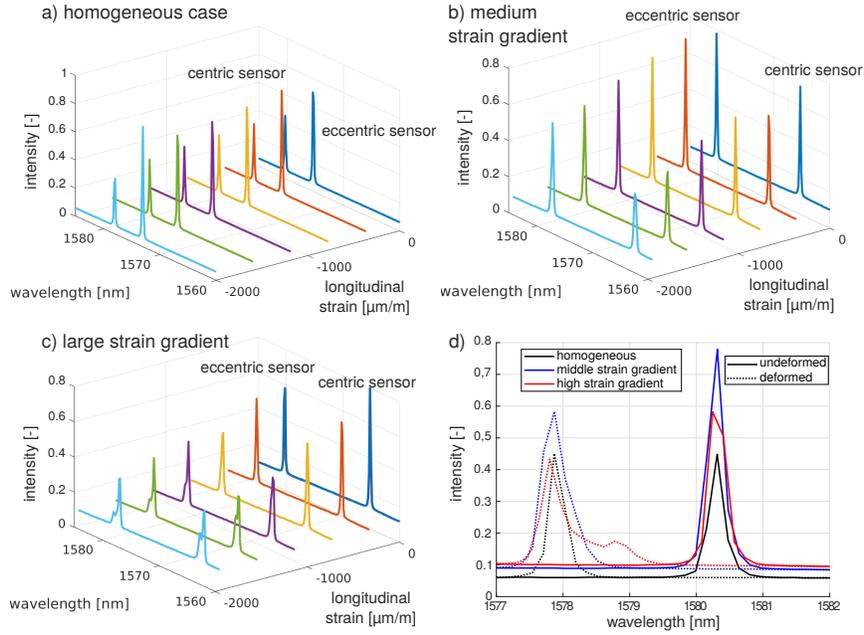}\caption{Demonstration of the different measure wavelength spectra over the loading (a, b, c). In part d, the displacement of the reflected peak and the different end formation is shown.}\label{fig:spektren_4er}
\end{figure} 

\begin{figure}\centering
\includegraphics[width=0.49\linewidth]{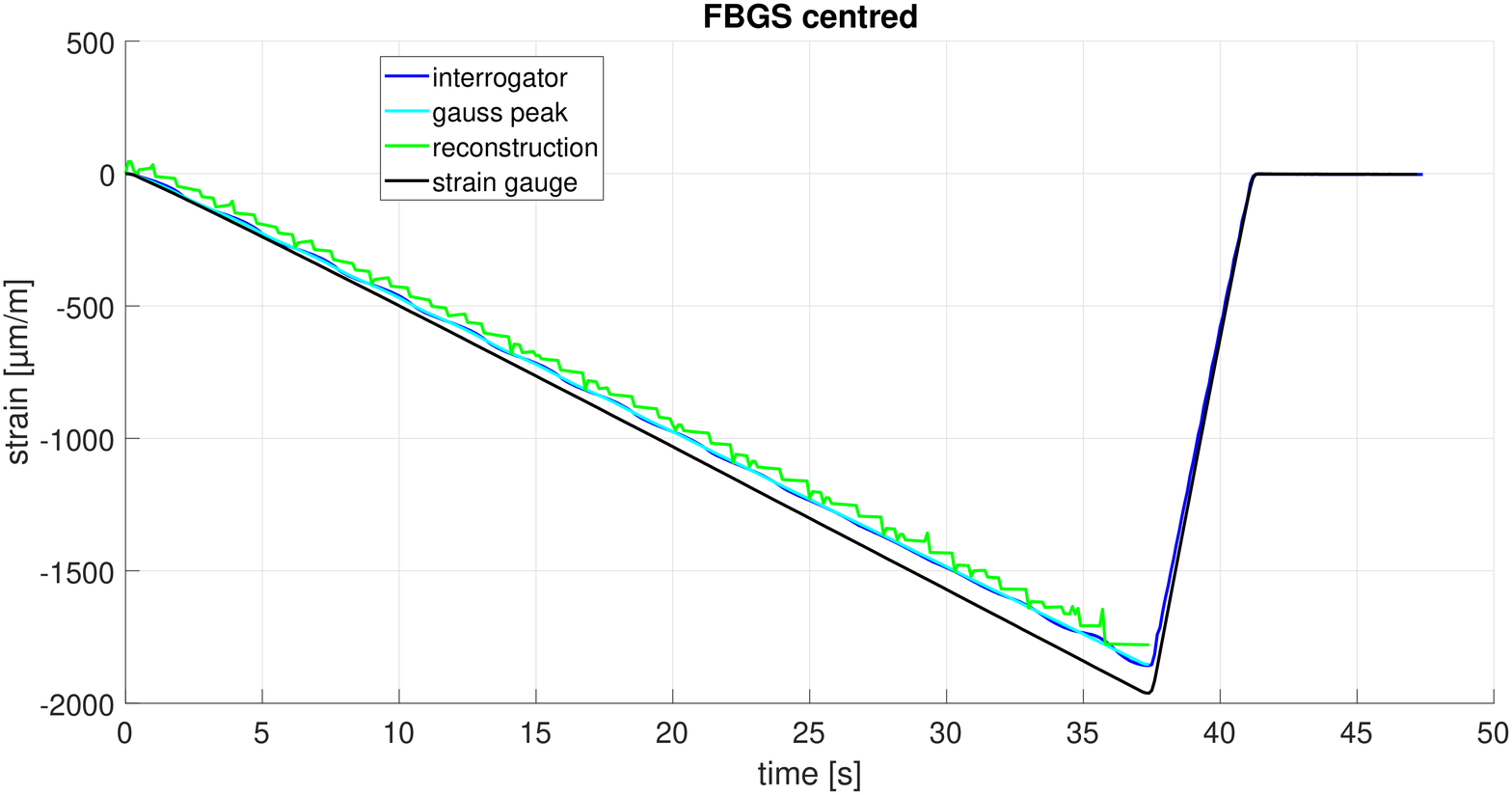}
\includegraphics[width=0.49\linewidth]{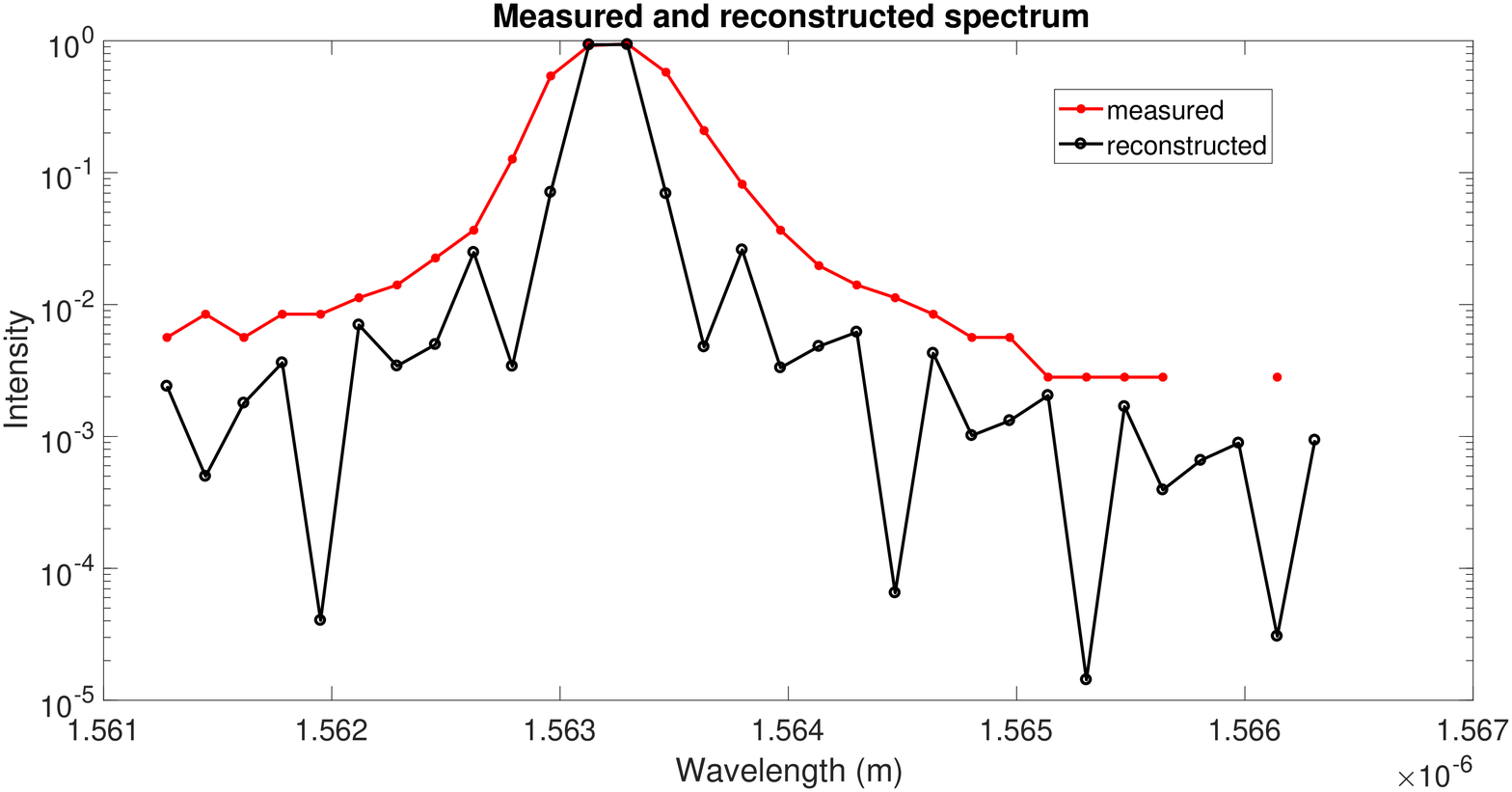}
\caption{Homogeneous model applied to a moderately inhomogeneous strain gradient (left). No evaluation of the FBG follows the strain gauge, the strain is systematically underestimated. 
The measured spectrum (right) and its reconstruction using a homogeneous model shows the model misfit.}
\label{fig:strain_mittel}
\end{figure}

Since the FBG is no longer homogeneous due to the applied force, it becomes necessary to resolve the strain over the FBG spatially. 
To this end, we employ the model of Section~\ref{ssec:chirp}. 
The presence of a strain gradient means that the period $\Lambda$ of the FBG refractive index distribution is no longer constant but varies locally, which we can model with a chirp function $\Phi$. 
The strain gauge only measures an average strain along the FBG, hence a different reference against which to compare our reconstruction is needed. 
To this end, we carry out a finite element simulation of the experiment using Abaqus CAE V6.14. 
The beam is modelled  with eight-node brick element (C3D8) with linear-elastic material behaviour with a Young's Modulus of $70\,GPa$ and a Poisson ratio of $0.3$. 
The beam is in contact with rolls as analytical rigid bodies. 
The contact behaviour is a hard contact in the normal direction without friction in the tangential direction. The strain analysis is obtained by way of Abaqus' path option. 
Different paths describe the three different sensors (centric and eccentric FBG sensor and electrical strain gauge). 
Using the paths we can resolve the strain along the each sensor, which we can use as a reference for our reconstruction. 
Note that this reference may itself deviate from the true strain since it stems from a simulation. 
In Figure~\ref{fig:mittelVG_lowreg} we show a reconstruction without regularization, and in Figure \ref{fig:mittelVG} the regularization with a hand-tuned regularization parameter is shown. 
The regularization is able to produce a smooth strain curve, but is observed to underestimate the strain magnitude. 
We remark that this may also be explained by the wavelength discretization and the resulting inaccuracy in the determination of the spectrum peak. 
In our setup, this inaccuracy may yield a strain offset of up to approximately $1.5*10^{-4} \frac{\mu m}{m}$, which would place the reconstruction right on top of the reference. 
Another explanation is that we are seeing the onset of birefringence. 
In Figure~\ref{fig:femvstrue} we compare a simulation of the FBG spectrum using the strain from the finite element simulation and the measurements. 
While the peak intensity of the measurements is lower than in the simulation, peak position and width of the peak agree well. 
Note that we used a much finer wavelength discretization in the simulation such that we can observe fine details in the sidelobes which are lost in the measurement due to poor wavelength resolution.

\begin{figure}
\includegraphics[width=\linewidth]{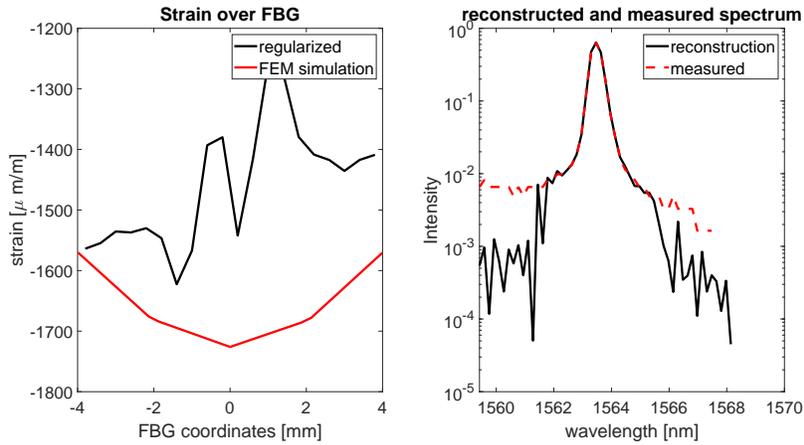}
\caption{Comparison of unregularized strain reconstruction with an FEM simulation (left); measured and reconstructed spectrum (right). As expected we obtain an oscillatory reconstruction, but are able to fit the measured spectrum closely.}
\label{fig:mittelVG_lowreg}
\end{figure}

\begin{figure}
\includegraphics[width=\linewidth]{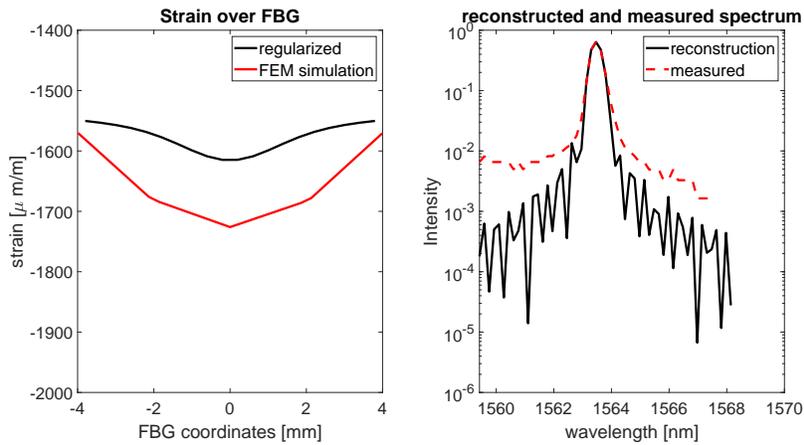}\caption{Comparison of regularized strain reconstruction with an FEM simulation (left); measured and reconstructed spectrum (right). 
The regularization achieves a smooth solution, however the reconstructed strain values are slightly damped.}
\label{fig:mittelVG}
\end{figure}

\begin{figure}
\includegraphics[width=0.49\linewidth]{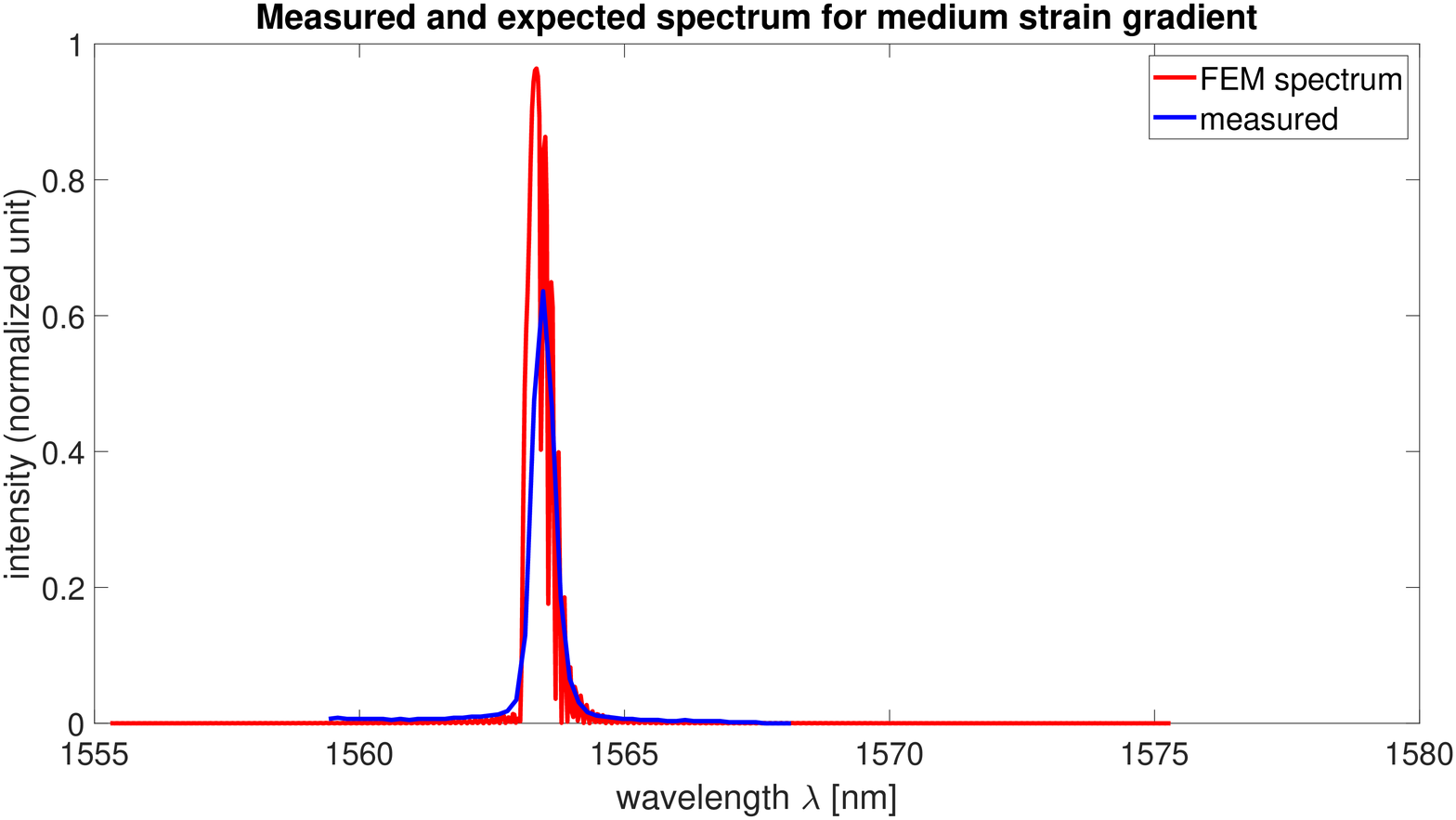}\includegraphics[width=0.49\linewidth]{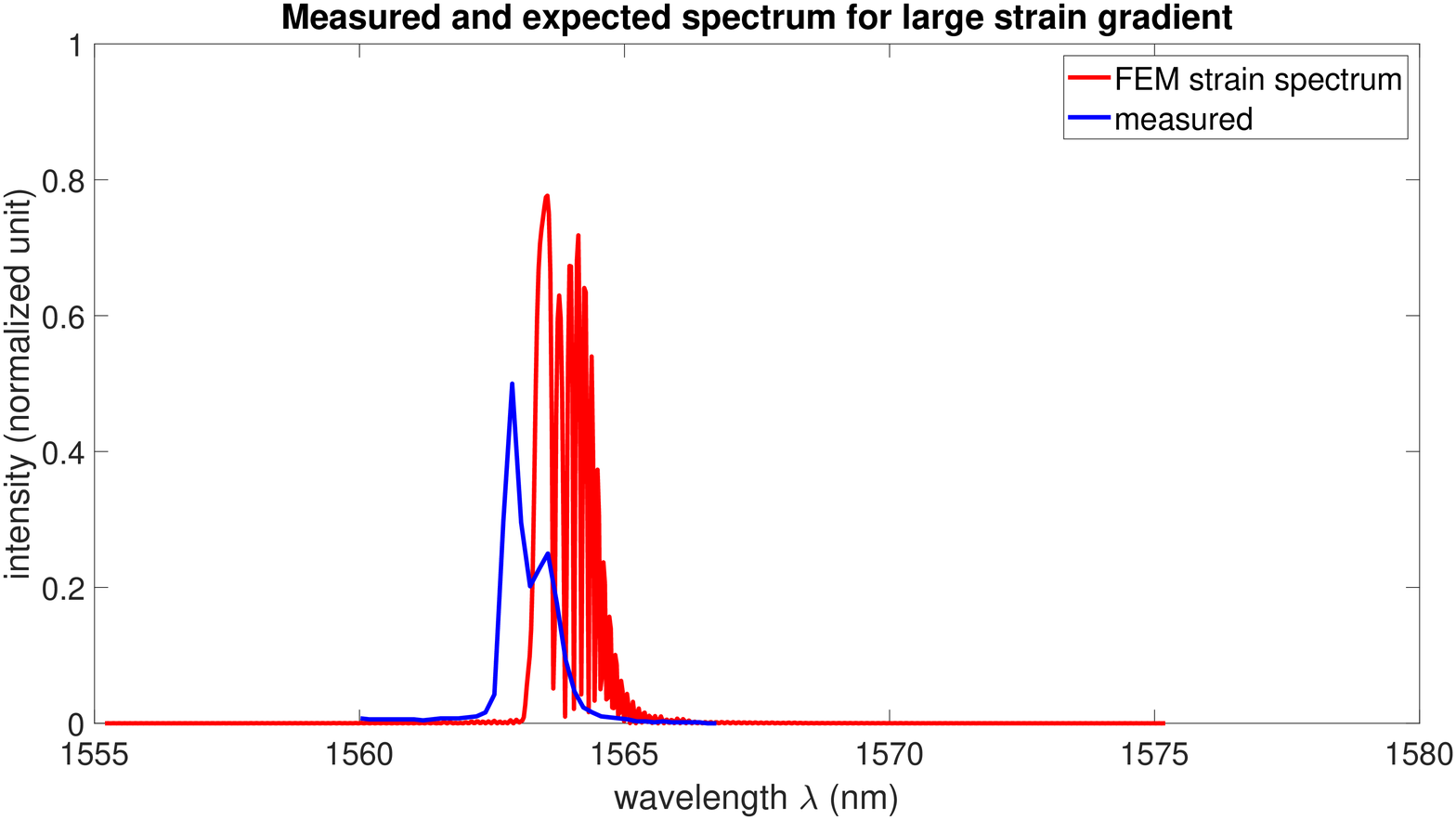}\caption{Measured spectra and spectra from FEM strain for inhomogeneous strain distributions. Left: medium strain gradient. Right: large strain gradient. For the medium strain, peak position and peak width match. For the large strain gradient, this is no longer the case, likely due to birefringence. Hence our model fails in the large strain gradient case.}\label{fig:femvstrue}
\end{figure}

\subsection{Analysis of Experimental Data: Large Strain Gradient} \label{ssec:high}

Our intention was to spatially resolve the strain along the FBG also for the large strain gradient case and compare the reconstruction to the finite element simulation. 
However, the reconstructions were not satisfying. 
We therefore simulated the expected spectrum given the strain obtained from the finite element simulation and compared it to the measured data, see Figure~\ref{fig:femvstrue}. 
It is easy to see that the spectra do not fit, i.e., the model must be incorrect. 
This is likely due to birefringence (see Remark~\ref{rem:biref}), since the large strain leads to significant transverse loading on the fibre, superimposed to the axial strain. 
It is the subject of future work to fully integrate this effect. 
A forward modelling without inversion can be found in \cite{PetersFEM}.

\subsection{Calibration of an Apodized FBG}

In our final experiment we leave the previous setting. 
By courtesy of Dr.\ Detleff Hofmann (German Federal Institute for Materials Research and Testing, BAM) we were provided with a high-resolution FBG spectrum. 
The only information provided was that the FBG is apodized (cf.\ Fig.~\ref{fig:FBGtypes}). 
None of the  characteristic FGB parameters were known, not even the grating length. 
As before, we simply guessed the necessary parameters to conduct the regularization as presented in Section \ref{ssec:apodized}. 
We set $L=6\,$mm and $n_0=1.46$, and calculated $\Lambda_0$ from the FBG peak $\lambda_B$. 
The spectra now have a wavelength resolution of $5\,$pm, i.e., a 33 times higher resolution than the previous set of measurements. 
The spectra are given in dB and range from $-10$dB to $-60$dB, i.e., can be considered essentially noise-free. 
In the reconstructions we would expect a parabolic curve for both $\dndc$ and $\dnac$. For $\dnac$ this is clearly the case. 
However, we were not able to find a smoother solution for $\dndc$. Intuitively we would attribute this to the unknown FBG-length $L$, since $\dndc$ shows strong artefacts at the boundaries. The reconstructions were surprisingly stable against variation of the regularization parameter, and most looked similar to that shown in Figure~\ref{fig:rec_bam}.

\begin{figure}
\includegraphics[width=\linewidth]{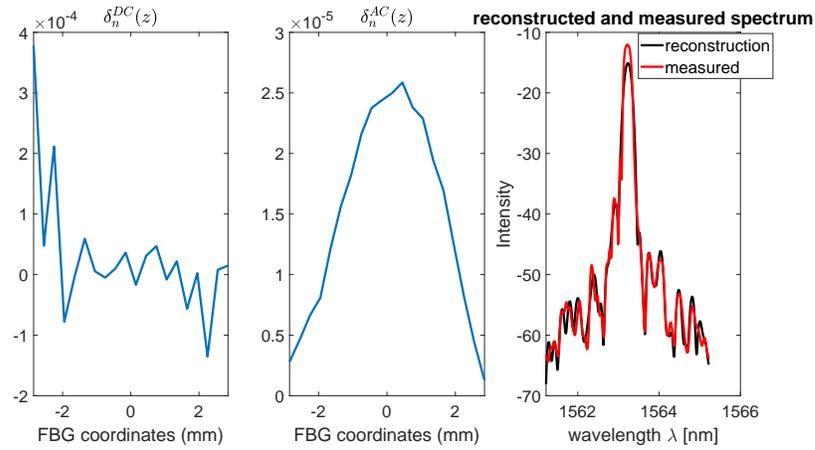}
\caption{Reconstruction of an apodized FBG from high-resolution measured data. 
We expected a parabolic curve of $\dndc$ and $\dnac$. The spectra match well.}
\label{fig:rec_bam}
\end{figure}

\section*{Conclusion}

We have presented the basic physics and modelling for the spectrum of FBG sensors and cast the recovery of the refractive index modulation from the FBG spectra as a nonlinear inverse problem.
We have discussed some sources of non-uniqueness in the reconstruction from measured FBG data and
demonstrated that regularization is able to reconstruct complex FBG configurations given that the parameters are chosen appropriately. 
We have also shown that choosing the parameters is difficult, and that conventional parameter selection rules fail. 
Therefore, the development of new ways of determining the regularization parameters is crucial but has to be postponed to a later publication. 
Finally, we evaluated our algorithm in experiments with real data. 
Despite the low data quality, we were able to find good approximations of the unknowns, which we verified against other experimental data or a finite element simulation measurements were not feasible. 
The paper demonstrates that regularization opens the possibility to extract significantly more data from FBG measurements than used in most contexts, opening this comparatively inexpensive and well-established technology to more involved experiments and practical uses.

\section*{Acknowledgement}
The work by S. Hannusch is a part of a measure, which is co-financed by tax revenue based on the budget approved by the members of the Saxon state parliament. D. Gerth was supported by Deutsche Forschungsgemeinschaft (DFG) under project GE3171/1-1. Financial support is gratefully acknowledged.

\end{document}